\documentclass[runningheads]{llncs}
\usepackage{graphicx}
\usepackage{fullpage}
\usepackage{times}
\usepackage{xcolor}
\usepackage{helvet}  
\usepackage{courier}  
\usepackage[hyphens]{url}  
\usepackage{graphicx} 
\usepackage[round]{natbib}  
\usepackage{caption} 
\usepackage{algorithm}
\usepackage{algorithmic}

\usepackage{comment}
\usepackage{booktabs}
\usepackage{amsfonts}
\usepackage{amsmath}
\usepackage{mathtools}
\usepackage[inline]{enumitem}
\usepackage{subfig}
\usepackage{tikz}
\usetikzlibrary{arrows}
\usepackage{newfloat}
\usepackage{listings}
\DeclareCaptionStyle{ruled}{labelfont=normalfont,labelsep=colon,strut=off} 
\floatstyle{ruled}
\newfloat{listing}{tb}{lst}{}
\floatname{listing}{Listing}

\newtheorem{numclaim}{Claim}

\newcommand{\N}{\mathbb{N}}

\newcommand{\abs}[1]{|#1|}
\newcommand{\upto}[1]{[#1]}

\newcommand{\Var}{\mathit{Var}}
\newcommand{\valFunc}[2]{V(#1,#2)}
\newcommand{\model}{v}

\newcommand{\prop}{\mathcal{P}}
\newcommand{\operators}{\Lambda}
\DeclareMathOperator{\lX}{\mathbf{X}}
\DeclareMathOperator{\lF}{\mathbf{F}}
\DeclareMathOperator{\lG}{\mathbf{G}}
\DeclareMathOperator{\lU}{\mathbf{U}}
\DeclareMathOperator{\limplies}{\rightarrow}
\newcommand{\lfalse}{\mathit{false}}
\newcommand{\ltrue}{\mathit{true}}

\newcommand{\dfa}{\mathcal{A}}
\newcommand{\states}{Q}
\newcommand{\final}{F}
\newcommand{\trans}[2]{\delta({#1,#2})}
\newcommand{\lang}[1]{L({#1})}

\newcommand{\positives}{P}
\newcommand{\negatives}{N}
\newcommand{\discardpile}{D}
\newcommand{\bound}{n}
\newcommand{\horizon}{K}
\newcommand{\nrep}{\bound\textrm{-description}}
\newcommand{\dfaencoding}{\Phi^{\dfa}}
\newcommand{\pref}[1]{\mathit{Pref}(#1)}
\newcommand{\ltlencoding}{\Psi^{\varphi,\negatives}}
\newcommand{\ltlencodingce}{\Omega^{\negatives,\discardpile}}

\newcommand{\ceDfa}{\texttt{CEG\textsubscript{DFA}}}
\newcommand{\hybridDfa}{\texttt{S-SYM\textsubscript{DFA}}}
\newcommand{\symbolicDfa}{\texttt{SYM\textsubscript{DFA}}}
\newcommand{\ceLtl}{\texttt{CEG\textsubscript{LTL}}}
\newcommand{\hybridLtl}{\texttt{S-SYM\textsubscript{LTL}}}
\newcommand{\symbolicLtl}{\texttt{SYM\textsubscript{LTL}}}

\immediate\write18{make -C resources/; cp -rTl resources/plots .} 

\begin{document}
\title{Learning Interpretable Temporal Properties from Positive Examples Only}
%
%
\author{
	Rajarshi Roy\inst{1}
	\and
	Jean-Rapha\"el Gaglione\inst{2}
	\and
	Nasim Baharisangari\inst{3}
	\and
	Daniel Neider\inst{4,5}\thanks{Part of this work has been conducted when the author was at the Carl von Ossietzky University of Oldenburg, Germany}
	\and
	Zhe Xu\inst{3}
	\and
	Ufuk Topcu\inst{2}
}
\authorrunning{R. Roy et al.}
%
\institute{
	Max Planck Institute for Software Systems, Kaiserslautern, Germany
	\and
	University of Texas at Austin, Texas, USA
	\and
	Arizona State University, Arizona, USA
	\and
	TU Dortmund University, Dortmund, Germany
	\and
	Center for Trustworthy Data Science and Security, University Alliance Ruhr, Germany
}
\maketitle              
\begin{abstract}
We consider the problem of explaining the temporal behavior of black-box systems using human-interpretable models.
Following recent research trends, we rely on the fundamental yet interpretable models of \emph{deterministic finite automata} (DFAs) and \emph{linear temporal logic} (LTL\textsubscript{f}) formulas.
In contrast to most existing works for learning DFAs and LTL\textsubscript{f} formulas, we consider learning from only positive examples.
Our motivation is that negative examples are generally difficult to observe, in particular, from black-box systems. 
To learn meaningful models from positive examples only, we design algorithms that rely on \emph{conciseness} and \emph{language minimality} of models as regularizers.
Our learning algorithms are based on two approaches: a symbolic and a counterexample-guided one. 
The symbolic approach exploits an efficient encoding of language minimality as a constraint satisfaction problem, whereas the counterexample-guided one relies on generating suitable negative examples to guide the learning.
Both approaches provide us with effective algorithms with minimality guarantees on the learned models.
To assess the effectiveness of our algorithms, we evaluate them on a few practical case studies.
\end{abstract}



\section{Introduction}

The recent surge of complex black-box systems in Artificial Intelligence has increased the demand for designing simple explanations of systems for human understanding.
Moreover, in several areas such as robotics, healthcare, and transportation~\citep{royalsociety19,scirobotics,molnar2022}, inferring human-interpretable models has become the primary focus to promote human trust in systems.

To enhance the interpretability of systems, we aim to explain their temporal behavior.
For this purpose, models that are typically employed include, among others, finite state machines and temporal logics~\citep{WeissGY18,0002FN20}.
Our focus is on two fundamental models: deterministic finite automata (DFAs)~\citep{RabinS}; and formulas in the de facto standard temporal logic: linear temporal logic (LTL)~\citep{Pnueli77}.
These models not only possess a host of desirable theoretical properties, but also feature easy-to-grasp syntax and intuitive semantics.
The latter properties make them particularly suitable as interpretable models with many applications, e.g., as task knowledge for robotic agents~\citep{KasenbergS17,Memarian00WT20}, as a formal specification for verification~\citep{LemieuxPB15}, as behavior classifier for unseen data~\citep{ShvoLIM21}, and several others~\citep{CamachoM19}.

The area of learning DFAs and LTL formulas is well-studied with a plethora of existing works (see related work).
Most of them tackle the typical binary classification problem~\citep{Gold78} of learning concise DFAs or LTL formulas from a finite set of examples partitioned into a positive and a negative set. 
However, negative examples are hard to obtain in my scenarios.
In safety-critical areas, often observing negative examples from systems (e.g., from medical devices and self-driving cars) can be unrealistic (e.g., by injuring patients or hitting pedestrians).
Further, often one only has access to a black-box implementation of the system and thus, can extract only its possible (i.e., positive) executions.

In spite of being relevant, the problem of learning concise DFAs and LTL formulas from positive examples, i.e., the corresponding \emph{one class classification} (OCC) problem, has garnered little attention.
The primary reason, we believe, is that, like most OCC problems, this problem is an ill-posed one.
Specifically, a concise model that classifies all the positive examples correctly is the trivial model that classifies all examples as positive.
This corresponds to a single state DFA or, in LTL, the formula $\ltrue$.
These models, unfortunately, convey no insights about the underlying system.

To ensure a well-defined problem, \citet{AvellanedaP18}, who study the OCC problem for DFAs, propose the use of the (accepted) \emph{language} of a model (i.e., the set of allowed executions) as a regularizer.
Searching for a model that has minimal language, however, results in one that classifies only the given examples as positive. 
To avoid this overfitting, they additionally impose an upper bound on the size of the model.
Thus, the OCC problem that they state is the following: given a set of positive examples $\positives$ and a size bound $\bound$, learn a DFA that 
\begin{enumerate*}[label=(\alph*)]
	\item classifies $\positives$ correctly, 
	\item has size at most $\bound$, and
	\item is language minimal.
\end{enumerate*}
For language comparison, the order chosen is set inclusion.

To solve this OCC problem, \citet{AvellanedaP18} then propose a \emph{counterexample-guided} algorithm.
This algorithm relies on generating suitable negative examples (i.e., counterexamples) iteratively to guide the learning process.
Since only the negative examples dictate the algorithm, in many iterations of their algorithm, the learned DFAs do not have a language smaller (in terms of inclusion) than the previous hypothesis DFAs.
This results in searching through several unnecessary DFAs.

To alleviate this drawback, our first contribution is a \emph{symbolic} algorithm for solving the OCC problem for DFA.
Our algorithm converts the search for a language minimal DFA symbolically to a series of satisfiability problems in Boolean propositional logic, eliminating the need for counterexamples.
The key novelty is an efficient encoding of the language inclusion check of DFAs in a propositional formula, which is polynomial in the size of the DFAs.
We then exploit an off-the-shelf SAT solver to check satisfiability of the generated propositional formulas and, thereafter, construct a suitable DFA.
We expand on this algorithm in Section~\ref{sec:dfa-learner}. 

We then present two novel algorithms for solving the OCC problem for formulas in LTL\textsubscript{f} (LTL over finite traces).
While our algorithms extend smoothly to traditional LTL (over infinite traces), our focus here is on LTL\textsubscript{f} due to its numerous applications in AI~\citep{GiacomoV13}.
Also, LTL\textsubscript{f} being a strict subclass of DFAs, the learning algorithms for DFAs cannot be applied directly to learn LTL\textsubscript{f} formulas.

Our first algorithm for LTL\textsubscript{f} is a \emph{semi-symbolic} algorithm, which combines ideas from both the symbolic and the counterexample-guided approaches.
Roughly, this algorithm exploits negative examples to overcome the theoretical difficulties of symbolically encoding language inclusion for LTL\textsubscript{f}. (LTL\textsubscript{f} inclusion check is known to be inherently harder than that for DFAs~\citep{SistlaC85}).
Our second algorithm is simply a counterexample-guided algorithm that relies solely on the generation of negative examples for learning.
Section~\ref{sec:ltl-learner} details both algorithms.

To further study the presented algorithms, we empirically evaluate them in several case studies.
We demonstrate that our symbolic algorithm solves the OCC problem for DFA in fewer (approximately one-tenth) iterations and runtime comparable to the counterexample-guided algorithm, skipping thousands of counterexample generations. 
Further, we demonstrate that our semi-symbolic algorithm solves the OCC problem for LTL\textsubscript{f} (in average) 
thrice as fast as
the counterexample-guided algorithm. 
All our experimental results can be found in Section~\ref{sec:experiments}.

%
%
%

\subsubsection{Related Work.}
The OCC problem described in this paper belongs to the body of works categorized as passive learning~\citep{Gold78}.
As alluded to in the introduction, in this topic, the most popular problem is the binary classification problem for learning DFAs and LTL formulas.
Notable works include the works by~\citet{BiermannF72,GrinchteinLP06,HeuleV10} for DFAs and~\citet{NeiderG18,CamachoM19,RahaRFN22} for LTL/LTL\textsubscript{f}.

The OCC problem of learning formal models from positive examples was first studied by~\citet{Gold67}.
This work showed that the exact identification  (in the limit) of certain models (that include DFAs and LTL\textsubscript{f} formulas) from positive examples is not possible.
Thereby, works have mostly focussed on models that are learnable easily from positive examples, such as pattern languages~\citep{Angluin198046}, stochastic finite state machines~\citep{CarrascoO99}, and hidden Markov Models~\citep{Stolcke92}.
None of these works considered learning DFAs or LTL formulas, mainly due to the lack of a meaningful regularizer.

Recently, \citet{AvellanedaP18} proposed the use of language minimality as a regularizer and, thereafter, developed an effective algorithm for learning DFAs.
While their algorithm cannot overcome the theoretical difficulties shown by~\citet{Gold67}, they still produce a DFA that is a concise description of the positive examples.
We significantly improve upon their algorithm by relying on a novel encoding of language minimality using propositional logic.

For temporal logics, there are a few works that consider the OCC problem. 
Notably, \citet{EhlersGN20} proposed a learning algorithm for a fragment of LTL which permits a representation known as universally very-weak automata (UVWs).
However, since their algorithm relies on UVWs, which has strictly less expressive power than LTL, it cannot be extended to full LTL.
Further, there are works on learning LTL~\citep{ChouOB22} and STL~\citep{JhaTSSS19} formulas from trajectories of high-dimensional systems.
These works based their learning on the assumption that the underlying system optimizes some cost functions.
Our method, in contrast, is based on the natural notion of language minimality to find tight descriptions, without any assumptions on the system.
There are some other works that consider the OCC problem for logics similar to temporal logic~\citep{Xuitl,Xugtl,SternJ17}.

A problem similar to our OCC problem is studied in the context of inverse reinforcement learning (IRL) to learn temporal rewards for RL agents from (positive) demonstrations.
For instance, \citet{KasenbergS17} learn concise LTL formulas that can distinguish between the provided demonstrations from random executions of the system.
To generate the random executions, they relied on a Markov Decision Process (MDP) implementation of the underlying system.
Our regularizers, in contrast, assume the underlying system to be a black-box and need no access to its internal mechanisms.
\citet{Vazquez-Chanlatte18} also learn LTL-like formulas from demonstrations. 
Their search required a pre-computation of the lattice of formulas induced by the subset order, which can be a bottleneck for scaling to full LTL.
Recently, \citet{HasanbeigJAMK21} devised an algorithm to infer automaton for describing high-level objectives of RL agents.
Unlike ours, their algorithm relied on user-defined hyper-parameters to regulate the degree of generalization of the inferred automaton.

\section{Preliminaries}
\label{sec:prelims}

In this section, we set up the notation for the rest of the paper.

Let $\N = \{1,2,\ldots\}$ be the set of natural numbers and $\upto{n}=\{1,2,\ldots,n\}$ be the set of natural numbers up to $n$.

\subsubsection{Words and Languages.}
To formally represent system executions, we rely on the notion of words, defined over a finite and nonempty \emph{alphabet} $\Sigma$.
The elements of $\Sigma$, which denote relevant system states, are referred to as \emph{symbols}.

A \emph{word} over $\Sigma$ is a finite sequence $w = a_1 \ldots a_n$ with $a_i \in \Sigma$, $i \in \upto{n}$.
The \emph{empty word} $\varepsilon$ is the empty sequence. 
The length $\abs{w}$ of $w$ is the number of symbols in it (note that $|\varepsilon| = 0$).
Moreover, $\Sigma^\ast$ denotes the set of all words over $\Sigma$.
Further, we use $w[i] = a_i$ to denote the $i$-th symbol of $w$ and $w[i{:}]=a_i\cdots a_{n}$ to denote the suffix of $w$ starting from position $i$.

A \emph{language} $L$ is any set of words from $\Sigma^\ast$.
We allow the standard set operations on languages such as $L_1\subseteq L_2$, $L_1\subset L_2$ and $L_1\setminus L_2$.

\subsubsection{Propositional logic.} 
All our algorithms rely on propositional logic and thus, we introduce it briefly.
Let $\mathit{Var}$ be a set of propositional variables, which take Boolean values $\{0,1\}$ ($0$ represents $\lfalse$, $1$ represents $\ltrue$). 
Formulas in propositional logic---which we denote by capital Greek letters---are defined recursively as 
\[\Phi \coloneqq x\in\Var \mid \neg \Phi \mid \Phi \lor \Phi.\]
As syntax sugar, we allow the following standard formulas: $\ltrue$, $\lfalse$, $\Phi\land\Psi$, $\Phi\rightarrow\Psi$, and $\Phi\leftrightarrow\Psi$. 

An \emph{assignment} $v\colon \Var \mapsto \{0,1\}$ maps propositional variables to Boolean values.
Based on an assignment $v$, we define the semantics of propositional logic using a valuation function $\valFunc{v}{\Phi}$, which is inductively defined as follows: 
\begin{align*}
\valFunc{v}{x}&=v(x)\\
\valFunc{v}{\neg\Psi}&= 1-\valFunc{v}{\Psi}\\ 
\valFunc{v}{\Psi\lor\Phi}&= \mathit{max}\{\valFunc{v}{\Psi} ,\valFunc{v}{\Phi}\}
\end{align*}
We say that $v$ satisfies $\Phi$ if  $\valFunc{v}{\Phi}=1$, and call $v$ a model of $\Phi$.
A formula $\Phi$ is satisfiable if there exists a model $v$ of $\Phi$.

Arguably, the most well-known problem in propositional logic---the satisfiability (SAT) problem---is the problem of determining whether a propositional formula is satisfiable or not.
With the rapid development of SAT solvers~\citep{sat-2021}, checking satisfiability of formulas with even millions of variables is feasible.
Most solvers can also return a model when a formula is satisfiable.

\section{Learning DFA from Positive Examples}
\label{sec:dfa-learner}

In this section, we present our symbolic algorithm for learning DFAs from positive examples. 
We begin by formally introducing DFAs.

A \emph{deterministic finite automaton} (DFA) is a tuple $\dfa = (\states,\Sigma,\delta,q_I,\final)$ where $\states$ is a finite set of states, $\Sigma$ is the alphabet, $q_I \in \states$ is the initial state,
$\final \subseteq \states$ is the set of final states, and $\delta \colon \states \times \Sigma \to \states$ is the transition function.
We define the size $\abs{\dfa}$ of a DFA as its number of states $\abs{Q}$.

Given a word $w=a_1 \ldots a_{n}\in\Sigma^\ast$, the run of $\dfa$ on $w$, denoted by $\dfa\colon q_1\xrightarrow{w}q_{n+1}$, is a sequence of states and symbols $q_1a_1q_2a_2\cdots a_{n}q_{n+1}$, such that $q_1=q_{I}$ and for each $i\in\upto{n}$, $q_{i+1} = \trans{q_i}{a_i}$.
Moreover, we say $w$ is accepted by $\dfa$ if the last state in the run $q_{n+1}\in\final$.
Finally, we define the language of $\dfa$ as $\lang{\dfa} = \{ w\in\Sigma^\ast~|~w\text{ is accepted by }\dfa \}$.

To introduce the OCC problem for DFAs, we first describe the learning setting. 
The OCC problem relies on a set of positive examples, which we represent using a finite set of words $\positives\subset\Sigma^\ast$.
Additionally, the problem requires a bound $\bound$ to restrict the size of the learned DFA.
The role of this size bound is two-fold:
\renewcommand\labelenumi{(\theenumi)}
\begin{enumerate*}
	\item it ensures that the learned DFA does not overfit $\positives$; and
	\item using a suitable bound, one can enforce the learned DFAs to be concise and, thus, interpretable.
\end{enumerate*}

Finally, we define a DFA $\dfa$ to be an $\nrep$ of $\positives$ if $\positives\subseteq \lang{\dfa}$ and $\abs{\dfa}\leq \bound$.
When $P$ is clear from the context, we simply say $\mathcal{A}$ is an $\nrep$.

We can now state the OCC problem for DFAs:
\begin{problem}[OCC problem for DFAs]\label{prob:occ-dfa}
	Given a set of positive words $\positives$ and a size bound $\bound$,
	learn a DFA $\dfa$ such that:
	\renewcommand\labelenumi{(\theenumi)}
	\begin{enumerate*}
		\item $\dfa$ is an $\nrep$; and
		\item for every DFA $\dfa'$ that is an $\nrep$, $L(\dfa')\not\subset L(\dfa)$.
	\end{enumerate*}
	
\end{problem}
Intuitively, the above problem asks to search for a DFA that is an $\nrep$ and has a minimal language.
Note that several such DFAs can exist since the language inclusion is a partial order on the languages of DFA.
We, here, are interested in learning only one such DFA, leaving the problem of learning all such DFAs as interesting future work.

\subsection{The Symbolic Algorithm}\label{sec:sym-dfa}
We now present our algorithm for solving Problem~\ref{prob:occ-dfa}.
Its underlying idea is to reduce the search for an appropriate DFA to a series of satisfiability checks of propositional formulas.
Each satisfiable propositional formula enables us to construct a guess, or a so-called \emph{hypothesis} DFA $\dfa$.
In each step, using the hypothesis $\dfa$, we construct a propositional formula $\Phi^{\dfa}$ to search for the next hypothesis $\dfa'$ with a language smaller (in the inclusion order) than the current one.
The properties of the propositional formula $\dfaencoding$ we construct are:
\renewcommand\labelenumi{(\theenumi)}
\begin{enumerate*}
	\item $\dfaencoding$ is satisfiable if and only if there exists a DFA $\dfa'$ that is an $\nrep$ and $\lang{\dfa'}\subset \lang{\dfa}$; and
	\item based on a model $\model$ of $\dfaencoding$, one can construct a such a DFA $\dfa'$.
\end{enumerate*}

Based on the main ingredient $\dfaencoding$, we design our learning algorithm as sketched in Algorithm~\ref{alg:symbolic-dfa}.
Our algorithm initializes the hypothesis DFA $\dfa$ to be $\dfa_{\Sigma^\ast}$, which is the one-state DFA that accepts all words in $\Sigma^\ast$.
Observe that $\dfa_{\Sigma^\ast}$ is trivially an $\nrep$, since $\positives\subset\Sigma^\ast$ and $\abs{\dfa_{\Sigma^\ast}}=1$.
The algorithm then iteratively exploits $\dfaencoding$ to construct the next hypothesis DFAs, until $\dfaencoding$ becomes unsatisfiable.
Once this happens, we terminate and return the current hypothesis $\dfa$ as the solution.
The correctness of this algorithm follows from the following theorem:
\begin{theorem}\label{thm:dfa-algo-correctness}
	Given positive words $\positives$ and a size bound $\bound$, Algorithm~\ref{alg:symbolic-dfa} always learns a DFA $\dfa$ that is an $\nrep$ and for every DFA $\dfa'$ that is an $\nrep$, $\lang{\dfa'}\not\subset\lang{\dfa}$.
\end{theorem}

\begin{algorithm}[tb]
	\caption{Symbolic Algorithm for Learning DFA}\label{alg:symbolic-dfa}
	\textbf{Input}: Positive words~$P$, bound~$\bound$
	\begin{algorithmic}[1]
		\STATE $\dfa\gets\dfa_{\Sigma^\ast}$, $\dfaencoding \coloneqq \Phi_{\mathtt{DFA}}\wedge\Phi_{\positives}$
		
		\WHILE{$\dfaencoding$ is satisfiable (with model $\model$)}
		\STATE $\dfa \gets$ DFA constructed from $v$
		\STATE $\dfaencoding \coloneqq \Phi_{\mathtt{DFA}}\wedge\Phi_{\positives} \wedge \Phi_{\subseteq\dfa} \wedge \Phi_{\not\supseteq\dfa}$\label{line:dfaencoding}
		\ENDWHILE
		\RETURN $\dfa$
	\end{algorithmic}
\end{algorithm}

We now expand on the construction of $\dfaencoding$. 
To achieve the aforementioned properties, we define $\dfaencoding$ as follows: 
\begin{align}\label{eq:dfa-encoding}
\dfaencoding \coloneqq \Phi_{\mathtt{DFA}} \wedge \Phi_{\positives} \wedge \Phi_{\subseteq\dfa} \wedge \Phi_{\not\supseteq\dfa}
\end{align}
The first conjunct $\Phi_{\mathtt{DFA}}$ ensures that the propositional variables we will use encode a valid DFA $\dfa'$. 
The second conjunct $\Phi_{\positives}$ ensures that $\dfa'$ accepts all positive words.
The third conjunct $\Phi_{\subseteq\dfa}$ ensures that $\lang{\dfa'}$ is a subset of $\lang{\dfa}$.
The final conjunct $\Phi_{\not\supseteq\dfa}$ ensures that $\lang{\dfa'}$ is not a superset of $\lang{\dfa}$.
Together, conjuncts $\Phi_{\subseteq\dfa}$ and $\Phi_{\not\supseteq\dfa}$  ensure that $\lang{\dfa'}$ is a proper subset of $\lang{\dfa}$.
In what follows, we detail the construction of each conjunct.

To encode the hypothesis DFA $\dfa'=(Q',\Sigma,\delta',{q'_I},\final')$ symbolically, following~\citet{HeuleV10}, we rely on the propositional variables:
\renewcommand\labelenumi{(\theenumi)}
\begin{enumerate*}
	\item $d_{p,a,q}$ where $p,q\in\upto{\bound}$ and $a\in\Sigma$; and
	\item $f_q$ where $q\in\upto{\bound}$.
\end{enumerate*}
The variables $d_{p,a,q}$ and $f_q$ encode the transition function $\delta'$ and the final states $F'$, respectively, of $\dfa'$.
Mathematically speaking, if $d_{p,a,q}$ is set to true, then $\delta'(p,a)=q$ and if $f_q$ is set to true, then $q\in\final'$.
Note that we identify the states $\states'$ using the set $\upto{n}$ and the initial state ${q'_{I}}$ using the numeral 1.

Now, to ensure $\dfa'$ has a deterministic transition function $\delta'$, $\Phi_{\mathtt{DFA}}$ asserts the following constraint:
\begin{align}
\bigwedge_{p\in\upto{\bound}}\bigwedge_{a\in\Sigma}\Big[\bigvee_{q\in\upto{\bound}} d_{p,a,q} \wedge \bigwedge_{q\neq q'\in\upto{\bound}} \big[\neg d_{p,a,q} \vee \neg  d_{p,a,q'}\big]\Big]\label{eq:deter}
\end{align}

Based on a model $v$ of the variables $d_{p,a,q}$ and $f_q$, we can simply construct $\dfa'$. 
We set $\delta'(p,a)$ to be the unique state $q$ for which $v(d_{p,a,q})=1$ and $q\in\final'$ if $v(f_q)=1$.  

Next, to construct conjunct $\Phi_{\positives}$, we introduce variables $x_{u,q}$ where $u\in\pref{P}$ and $q\in\upto{\bound}$, which track the run of $\dfa'$ on all words in $\pref{P}$, which is the set of prefixes of all words in $\positives$.
Precisely, if $x_{u,q}$ is set to true, then there is a run of $\dfa'$ on $u$ ending in state $q$, i.e., $\dfa'\colon {q'_{I}}\xrightarrow{u}q$.

Using the introduced variables,  $\Phi_{\positives}$ ensures that the words in $\positives$ are accepted by imposing the following constraints:
\begin{align}
x_{\varepsilon,1} \wedge \bigwedge_{q\in\{2,\ldots,\bound\}} \neg x_{\varepsilon,q}\label{eq:init_state}\\
\bigwedge_{ua\in\pref{P}} \bigwedge_{p,q\in\upto{\bound}} [x_{u,p} \wedge d_{p,a,q}]\limplies x_{ua,q}\label{eq:transtion}\\
\bigwedge_{w\in\positives} \bigwedge_{q\in \upto{\bound}} x_{w,q} \limplies f_q\label{eq:final-state}
\end{align}
The first constraint above ensures that the runs start in the initial state ${q'_I}$ (which we denote using 1) while the second constraint ensures that they adhere to the transition function.
The third constraint ensures that the run of $\dfa'$ on every $w\in\positives$ ends in a final state and is, hence, accepted.

For the third conjunct $\Phi_{\subseteq\dfa}$, we must track the synchronized runs of the current hypothesis $\dfa$ and the next hypothesis $\dfa'$ to compare their behavior on all words in $\Sigma^\ast$. 
To this end, we introduce auxiliary variables, $y^{\dfa}_{q,q'}$ where $q,q'\in\upto{\bound}$.
Precisely, $y^{\dfa}_{q,q'}$ is set to true, if there exists a word $w\in\Sigma^\ast$ such that there are runs $\dfa\colon q_I\xrightarrow{w}q$ and $\dfa'\colon{q'_I}\xrightarrow{w} q'$.

To ensure $\lang{\dfa'}\subseteq\lang{\dfa}$, $\Phi_{\subseteq\dfa}$ imposes the following constraints:
\begin{align}
y^{\dfa}_{1,1}\label{eq:sync-init}\\
\bigwedge_{q=\trans{p}{a}} \bigwedge_{p',q'\in\upto{\bound}}\bigwedge_{a\in\Sigma} \Big[\big[y^{\mathcal{A}}_{p,p'} \wedge d_{p',a,q'}\big] \rightarrow y^{\mathcal{A}}_{q,q'}\Big]\label{eq:sync-transition}\\
\bigwedge_{p\not\in F}\bigwedge_{p'\in\upto{\bound}}  \Big[y^{\mathcal{A}}_{p,p'} \rightarrow \neg f_{p'}\Big]\label{eq:sync-final}
\end{align}
The first constraint ensures that the synchronized runs of $\dfa$ and $\dfa'$ start in the respective initial states, while the second constraint ensures that they adhere to their respective transition functions.
The third constraint ensures that if the synchronized run ends in a non-final state in $\dfa$, it must also end in a non-final state in $\dfa'$, hence forcing $L(\dfa')\subseteq L(\dfa)$.

For constructing the final conjunct $\Phi_{\not\supset\dfa}$, the variables we exploit rely the following result:
\begin{lemma}\label{lem:time-horizon}
	Let $\dfa$, $\dfa'$ be DFAs such that $\abs{\dfa}=\abs{\dfa'}=\bound$ and $\lang{\dfa'}\subset\lang{\dfa}$, and let $\horizon=\bound^2$. Then there exists a word $w\in\Sigma^\ast$ such that $|w|\leq\horizon$ and $w\in\lang{\dfa}\setminus\lang{\dfa'}$. 
\end{lemma}
This result provides an upper bound to the length of a word that can distinguish between DFAs $\dfa$ and $\dfa'$ 

Based on this result, we introduce variables $z_{i,q,q'}$ where $i\in\upto{\bound^2}$ and $q,q'\in\upto{\bound}$ to track the
synchronized run of $\dfa$ and $\dfa'$ on a word of length at most $\horizon=n^2$.
Precisely, if $z_{i,q,q'}$ is set to true, then there exists a word $w$ of length $i$, with the runs $\dfa\colon q_I\xrightarrow{w}q$ and $\dfa':q_{I}'\xrightarrow{w} q'$. 

Now, $\Phi_{\not\supseteq\dfa}$ imposes the following constraints:
\begin{align}
z_{0,1,1}\label{eq:one-word-init}\\
\bigwedge_{i\in\upto{\bound^2}}\Big[\bigvee_{q,q'\in\upto{\bound}} z_{i,q,q'} \wedge \big[\bigwedge_{\substack{p\neq q\in\upto{\bound}\\p'\neq q' \in\upto{\bound}}}\neg z_{i,p,p'} \vee \neg  z_{i,q,q'}\big]\Big]\label{eq:one-word-unique}\\
\bigwedge_{\substack{p,q\in\upto{\bound}\\p',q'\in\upto{\bound}}} \Big[ \big[z_{i,p,p'} \wedge z_{i+1,q,q'}\big] \limplies \bigvee_{\substack{a\in\Sigma \text{ where}\\{q=\delta(p,a)}}} d_{p',a,q'} \Big]\label{eq:one-word-transition}\\
\bigvee_{i\in\upto{\bound^2}}\bigvee_{\substack{q\in\final\\q'\in\upto{\bound}}} \Big[z_{i,q,q'} \wedge \neg f_{q'}\Big]\label{eq:one-word-final}
\end{align}
The first three constraints above ensure that words up to length $\bound^2$ have a valid synchronised run on the two DFAs $\dfa$ and $\dfa'$.
The final constraint ensures that there is a word of length $\leq\bound^2$ on which the synchronized run ends in a final state in $\dfa$ but not in $\dfa'$, ensuring $L(\dfa)\not\subseteq L(\dfa')$.

We now prove the correctness of the propositional formula $\dfaencoding$ that we construct using the above constraints:
\begin{theorem}\label{thm:dfa-encoding-correctness}
	Let $\dfaencoding$ be as defined above. Then, we have the following:
	\begin{enumerate}
		\item If $\dfaencoding$ is satisfiable, then there exists a DFA $\dfa'$ that is an $\nrep$ and $\lang{\dfa'}\subset \lang{\dfa}$.
		\item If there exists a DFA $\dfa'$ that is an $\nrep$ and $\lang{\dfa'}\subset \lang{\dfa}$, then $\dfaencoding$ is satisfiable. 
	\end{enumerate}
\end{theorem}

To prove the above theorem, we propose intermediate claims, all of which we prove first.
For the proofs, we assume $\dfa=(\states,\Sigma,\delta,{q_I},\final)$ to be the current hypothesis, $v$ to be a model of $\dfaencoding$, and $\dfa'=(\states',\Sigma,\delta',{q'_I},\final')$ to be the DFA constructed from the model $v$ of $\dfaencoding$.
\setlist[description]{font=\normalfont\itshape\space}
\begin{numclaim}\label{cl:positives}
	For all $u\in\pref{\positives}$, $\dfa':{q'_I}\xrightarrow{u}q$ implies $v(x_{u,q})=1$.
\end{numclaim}
\begin{proof}
	We prove the claim using induction on the length $\abs{u}$ of the word $u$.
	\begin{description}
		\item[Base case:] Let $u=\varepsilon$. Based on the definition of runs, $\dfa':{q'_I}\xrightarrow{\epsilon}q$ implies $q={q'_I}$.
		Also, using Constraint~\ref{eq:init_state}, we have $v(x_{\varepsilon,q})=1$ if and only if $q={q'_I}$ (note $q'_I$ is indicated using numeral 1).
		Combining these two facts proves the claim for the base case.
		
		\item[Induction step:] As induction hypothesis, let $\dfa':{q_I}'\xrightarrow{u}q$ implies $v(x_{u,q})=1$ for all words $u\in\pref{\positives}$ of length $\leq k$.
		Now, consider the run $\dfa':{q'_I}\xrightarrow{u} q \xrightarrow{a} q'$ for some $ua\in\pref{\positives}$.
		For this run, based on the induction hypothesis and the construction of $\dfa'$, we have $v(x_{u,q})=1$ and $v(d_{p,a,q})=1$.
		Now, using Constraint~\ref{eq:transtion}, $v(x_{u,q})=1$ and $v(d_{p,a,q})=1$ implies $v(x_{ua,q})=1$, thus, proving the claim.
	\end{description}
	\qed
\end{proof}

\begin{numclaim}\label{cl:inclusion}
	For all $w\in \Sigma^\ast$, $\dfa:q_I\xrightarrow{w}q$ and $\dfa': {q'_I}\xrightarrow{w}q'$ imply $v(y^{\dfa}_{q,q'})=1$.
\end{numclaim}
\begin{proof}
	We prove this using induction of the length $\abs{w}$ of the word $w$.
	\begin{description}
		\item[Base case:] Let $w=\varepsilon$. Based on the definition of runs, $\dfa:{q_I}\xrightarrow{\varepsilon}q$, $\dfa':{q'_I}\xrightarrow{\varepsilon}q'$ implies $q={q_I}$ and $q={q'_I}$.
		Also, using Constraint~\ref{eq:sync-init}, $q=q_I$ and $q'={q'_I}$ imply $v(y^{\dfa}_{q,q'})=1$ (note $q_I$ and $q'_I$ are both indicated using numeral 1).
		Combining these two facts proves the claim for the base case.
		
		\item[Induction step:] As induction hypothesis, let $\dfa:{q_I}\xrightarrow{w}q$ and $\dfa':{q'_I}\xrightarrow{w}q'$ imply $v(y^{\dfa}_{q,q'})=1$ for all words $w\in\Sigma^\ast$ of length $\leq k$.
		Now, consider the runs $\dfa:{q_I}\xrightarrow{w} p \xrightarrow{a} q$ and  $\dfa':{q_I}'\xrightarrow{w} p' \xrightarrow{a} q'$ for some word $wa\in\Sigma^\ast$
		For these runs, based on the induction hypothesis and the construction of $\dfa'$, we have $v(y^{\dfa}_{p,p'})=1$ and $v(d_{p',a,q'})=1$.
		Now, using Constraint~\ref{eq:sync-transition}, we can say that $v(y^{\dfa}_{p,p'})=1$ and $v(d_{p',a,q'})=1$ imply $v(y^{\dfa}_{q,q'})=1$ (where $q=\delta(p,a)$), thus, proving the claim.
	\end{description}
\qed	
\end{proof}

\begin{numclaim}\label{cl:proper-inclusion}
	$v(z_{i,q,q'})=1$ implies there exists $w\in \Sigma^i$ with runs $\dfa:q_I\xrightarrow{w}q$ and $\dfa': {q'_I}\xrightarrow{w}q'$.
\end{numclaim}
\begin{proof}
	We prove this using  induction on the parameter $i$.
	\begin{description}
		\item[Base case:] Let $i=0$. Based on the Constraints~\ref{eq:one-word-init} and~\ref{eq:one-word-unique}, $v(z_{0,q,q'})=1$ implies $q={q_I}$ and $q={q_I}'$.
		Now, there always exists a word of length $0$, $w=\varepsilon$, for which $\dfa:q_I\xrightarrow{\varepsilon}q$ and $\dfa': {q'_I}\xrightarrow{\varepsilon}q'$ proving the claim for the base case.
		
		\item[Induction step:] As induction hypothesis, let $v(z_{k,p,p'})=1$ and thus, $w$ be a word of length $k$ such that $\dfa:q_I\xrightarrow{w}p$ and $\dfa': {q'_I}\xrightarrow{w}p'$.
		Now, assume $v(z_{k+1,q,q'})=1$.
		Based on Constraint~\ref{eq:one-word-transition}, for some $a\in\Sigma$ such that $q=\delta(p,a)$, $v(d_{p',a,q'})=1$.
		Thus, on the word $wa$, there are runs $\dfa:q_I\xrightarrow{w}p\xrightarrow{a}q$ and $\dfa': {q'_I}\xrightarrow{w}p'\xrightarrow{a}q'$, proving the claim.
	\end{description}
\qed
\end{proof}

We are now ready to prove Theorem~\ref{thm:dfa-encoding-correctness}, i.e., the correctness of $\dfaencoding$.

\begin{proof}[of Theorem~\ref{thm:dfa-encoding-correctness}]
	For the forward direction, consider that $\dfaencoding$ is satisfiable with a model $v$ and $\dfa'$ is the DFA constructed using the model $v$.
	First, using Claim~\ref{cl:positives}, $\dfa'\colon {q'_I}\xrightarrow{w} q$ implies $v(x_{w,q})=1$ for all $w\in\positives$.
	Now, based on Constraint~\ref{eq:final-state}, $v(x_{w,q})=1$ implies $v(f_q)=1$ for all $w\in\positives$.
	As a result, for each $w\in\positives$, its run $\dfa'\colon {q'_I}\xrightarrow{w} q$ must end in a final state $q\in\final'$ in $\dfa'$.
	Thus, $\dfa'$ accepts all positive words and hence, is an $\nrep$.
	Second, using Claim~\ref{cl:inclusion}, $\dfa:q_I\xrightarrow{w}q$ and $\dfa': {q'_I}\xrightarrow{w}q'$ imply $v(y^{\dfa}_{q,q'})=1$ for all words $w\in\Sigma^\ast$.
	Thus, based on Constraint~\ref{eq:sync-final}, if $q\not\in\final$, then $q'\not\in\final$, implying $\lang{\dfa'}\subseteq\lang{\dfa}$.
	Third, using Claim~\ref{cl:proper-inclusion}, $v(z_{i,q,q'})=1$ implies that there exists $w\in \Sigma^i$ with runs $\dfa:q_I\xrightarrow{w}q$ and $\dfa': {q'_I}\xrightarrow{w}q'$.
	Now, based on Constraint~\ref{eq:one-word-final}, there exists some $i\leq \bound^2$, $q\in F$ and $q'$ in $\dfa'$ such that $v(z_{i,q,q'})=1$ and $v(f_{q'})=0$.
	Combining this fact with Claim~\ref{cl:proper-inclusion}, we deduce that there exists $w\in\Sigma^\ast$ with length $\leq\bound^2$ with run $\dfa:q_I\xrightarrow{w}q$ ending in final state $q\in\final$ and run $\dfa': {q'_I}\xrightarrow{w}q'$ not ending in a final state $q'\in\final'$.
	This shows that $\lang{\dfa}\neq\lang{\dfa'}$.
	We thus conclude that $\dfa$ is an $\nrep$ and $\lang{\dfa'}\subset\lang{\dfa}$.
	
	For the other direction, based on a suitable DFA $\dfa'$, we construct an assignment $v$ for all the introduce propositional variables. 
	First, we set $v(d_{p,a,q})=1$ if $\delta'(p,a)=q$ and $v(f_q)=1$ if $q\in\final'$.
	Since $\delta'$ is a deterministic function, $v$ satisfies the Constraint~\ref{eq:deter}.
	Similarly, we set $v(x_{u,q})=1$ if $\dfa'\colon {q'_{I}}\xrightarrow{u}q$ for some $u\in\pref{\positives}$.
	It is a simple exercise to check that $v$ satisfies Constraints~\ref{eq:init_state} to~\ref{eq:final-state}.
	Next, we set $v(y^{\dfa}_{q,q'})=1$ if there are runs $\dfa\colon q_I\xrightarrow{w}q$ and $\dfa'\colon {q'_I}\xrightarrow{w}q'$ on some word $w\in\Sigma^\ast$.
	Algorithmically, we set $v(y^\dfa_{q,q'})=1$ if states $q$ and $q'$ are reached in the synchronized run of $\dfa$ and $\dfa'$ on some word (which is typically computed using a breadth-first search on the product DFA). 
	It is easy to see that such an assignment $v$ satisfies Constraints~\ref{eq:sync-init} to~\ref{eq:sync-final}.
	Finally, we set assignment to $z_{i,q,q'}$ exploiting a word $w$ which permits runs $\dfa\colon q_I\xrightarrow{w}q$ and $\dfa'\colon {q'_I}\xrightarrow{w}q'$, where $q$ is in $\final$ but $q'$ not in $\final'$.  
	In particular, we set $v(z_{i,q,q'})=1$ for $i=\abs{u}$ and $\dfa\colon q_I\xrightarrow{u}q$ and $\dfa'\colon {q'_I}\xrightarrow{u}q'$ for all prefixes $u$ of $w$.
	Such an assignment encodes a synchronized run of the DFAs $\dfa$ and $\dfa'$ on word $w$ that ends in a final state in $\dfa$, but not in $\dfa'$. 
	Thus, $v$ satisfies Constraints~\ref{eq:one-word-init} to~\ref{eq:one-word-final}.
	\qed
\end{proof}

\begin{proof}[of Theorem~\ref{thm:dfa-algo-correctness}]
First, observe that Algorithm~\ref{alg:symbolic-dfa} always terminates. 
This is because, there is finitely many $\nrep$ for a given $\bound$ (following its definition) and in each iteration, the algorithm finds a new one.
Next, observe that the algorithm terminates when $\dfaencoding$ is unsatisfiable.
Now, based on the properties of $\dfaencoding$ established in Theorem~\ref{thm:dfa-encoding-correctness}, if $\dfaencoding$ is unsatisfiable for some DFA $\dfa$, then there are no $\nrep$ DFA $\dfa'$ for which $L(\dfa')\subset L(\dfa)$, thus, proving the correctness of the algorithm.
\qed
\end{proof}

\section{Learning LTL\textsubscript{f} formulas from Positive Examples}
\label{sec:ltl-learner}

We now switch our focus to algorithms for learning LTL\textsubscript{f} formulas from positive examples. 
We begin with a formal introduction to LTL\textsubscript{f}.

\emph{Linear temporal logic} (over finite traces) (LTL\textsubscript{f}) is a logic that reasons about the temporal behavior of systems using temporal modalities.
While originally LTL\textsubscript{f} is built over propositional variables $\prop$,
to unify the notation with DFAs, we define LTL\textsubscript{f} over an alphabet $\Sigma$.
It is, however, not a restriction since an LTL\textsubscript{f} formula over $\prop$ can always be translated to an LTL\textsubscript{f} formula over $\Sigma=2^\prop$. 
Formally, we define LTL\textsubscript{f} formulas---denoted by Greek small letters---inductively as:
\[ \varphi \coloneqq a\in\Sigma \mid \lnot\varphi \mid \varphi \lor \varphi \mid \lX\varphi \mid \varphi\lU\varphi \]
As syntactic sugar, along with additional constants and operators used in propositional logic, we allow the standard temporal operators $\lF$ (``finally'') and $\lG$ (``globally'').
We define $\operators = \{ \lnot, \lor, \land, \limplies, \lX, \lU, \lF, \lG \}~\cup~\Sigma$ to be the set of all operators (which, for simplicity, also includes symbols).
We define the size $\abs{\varphi}$ of $\varphi$ as the number of its unique subformulas; e.g., size of $\varphi=(a \lU \lX b) \lor \lX b$ is five, since its five distinct subformulas are 
$a, b, \lX b, a \lU \lX b$, and $(a \lU \lX b) \lor \lX b$.

To interpret LTL\textsubscript{f} formulas over (finite) words, we follow the semantics proposed by~\citet{GiacomoV13}.
Given a word $w$, we define recursively when a LTL\textsubscript{f} formula holds at position $i$, i.e., $w,i\models \varphi$, as follows: 
\begin{align*}
&w,i \models a \in \Sigma \text{ if and only if }a = w[i]\\
&w,i \models \lnot \varphi \text{ if and only if } w,i\not\models \varphi \\
&w,i \models \lX \varphi \text{ if and only if } i < |w| \text{ and } w, i+1 \models \varphi \\
&w,i \models \varphi \lU \psi \text{ if and only if } w,j \models \psi \text{ for some }\nonumber \\
&\hspace{10mm}i \leq j \leq |w| \text{ and } w,i' \models \varphi\text{ for all }i \leq i' < j
\end{align*}
We say $w$ \emph{satisfies} $\varphi$ or, alternatively, $\varphi$ \emph{holds on} $w$ if $w,0\models\varphi$, which, in short, is written as $w\models \varphi$.


The OCC problem for LTL\textsubscript{f} formulas, similar to Problem~\ref{prob:occ-dfa}, relies upon a set of positive words $\positives\subset\Sigma^\ast$ and a size upper bound $\bound$.
Moreover, an LTL\textsubscript{f} formula $\varphi$ is an $\nrep$ of $\positives$ if, for all $w\in\positives$, $w\models\varphi$, and $\abs{\varphi}\leq\bound$.
Again, we omit $\positives$ from $\nrep$ when clear.
Also, in this section, an $\nrep$ refers only to an LTL formula.

We state the OCC problem for LTL\textsubscript{f} formulas as follows: 
\begin{problem}[OCC problem for LTL\textsubscript{f} formulas]\label{prob:occ-ltl}
	Given a set of positive words $\positives$ and a size bound $\bound$, learn an LTL\textsubscript{f} formula $\varphi$ such that:
	\renewcommand\labelenumi{(\theenumi)}
	\begin{enumerate*}
		\item $\varphi$ is an $\nrep$; and
		\item for every LTL\textsubscript{f} formula $\varphi'$ that is an $\nrep$, $\varphi'\not\rightarrow\varphi$ or $\varphi\rightarrow\varphi'$.
	\end{enumerate*}
\end{problem}
Intuitively, the above problem searches for an LTL\textsubscript{f} formula $\varphi$ that is an $\nrep$ and holds on a minimal set of words.
Once again, like Problem~\ref{prob:occ-dfa}, there can be several such LTL\textsubscript{f} formulas, but we are interested in learning exactly one.

\subsection{The Semi-Symbolic Algorithm}\label{sec:semi-sym-ltl}
Our \emph{semi-}symbolic, does not only depend on the current hypothesis, an LTL\textsubscript{f} formula $\varphi$ here, as was the case in Algorithm~\ref{alg:symbolic-dfa}.
In addition, it relies on a set of negative examples $\negatives$, accumulated during the algorithm.
Thus, using both the current hypothesis $\varphi$ and the negative examples $\negatives$, we construct a propositional formula $\ltlencoding$ to search for the next hypothesis $\varphi'$.
Concretely, $\ltlencoding$ has the properties that:
\begin{enumerate*}
	\item $\ltlencoding$ is satisfiable if and only if there exists an LTL\textsubscript{f} formula $\varphi'$ that is an $\nrep$, does not hold on any $w\in N$, and $\varphi\not\rightarrow\varphi'$; and
	\item based on a model $\model$ of $\ltlencoding$, one can construct such an LTL\textsubscript{f} formula $\varphi'$.
\end{enumerate*}

\begin{algorithm}[tb]
	\caption{Semi-symbolic Algorithm for learning LTL\textsubscript{f} formula }\label{alg:hybrid-ltl}
	\textbf{Input}: Positive words~$P$, bound~$\bound$
	\begin{algorithmic}[1]
		\STATE $N\gets\emptyset$ 
		\STATE $\varphi\gets\mathit{true}$, $\ltlencoding \coloneqq \Psi_{\mathtt{LTL}}\wedge\Psi_{\positives}$
		\WHILE{$\ltlencoding$ is satisfiable (with model $\model$)}
		\STATE $\varphi' \gets$ LTL formula constructed from $v$
		\IF{$\varphi'\rightarrow\varphi$}
		\STATE Update $\varphi$ to $\varphi'$\label{line:subset}
		\ELSE
		\STATE Add $w$ to $\negatives$, where $w\models\neg\varphi'\wedge\varphi$\label{line:incomparable}
		\ENDIF
		\STATE $\ltlencoding \coloneqq \Psi_{\mathtt{LTL}}\wedge \Psi_{\positives}\wedge\Psi_{\negatives}\wedge\Psi_{\not\leftarrow \varphi}$
		\ENDWHILE
		\RETURN $\varphi$
	\end{algorithmic}
\end{algorithm}
The semi-symbolic algorithm whose pseudocode is presented in Algorithm~\ref{alg:hybrid-ltl} follows a paradigm similar to the one illustrated in Algorithm~\ref{alg:symbolic-dfa}.
However, unlike the previous algorithm, the current guess $\varphi'$, obtained from a model of $\ltlencoding$, may not always satisfy the relation $\varphi'\rightarrow\varphi$.
In such a case, we generate a word (i.e., a negative example) that satisfies $\varphi$, but not $\varphi'$ to eliminate $\varphi'$ from the search space.
For the generation of words, we rely on constructing DFAs from the LTL\textsubscript{f} formulas~\citep{ZhuTLPV17} and then performing a breadth-first search over them.
If, otherwise, $\varphi'\rightarrow\varphi$, we then update our current hypothesis and continue until $\ltlencoding$ is unsatisfiable.
Overall, this algorithm learns an appropriate LTL\textsubscript{f} formula with the following guarantee:
\begin{theorem}\label{thm:semi-symbolic-ltl-correctness}
	Given positive words $\positives$ and size bound $\bound$, Algorithm~\ref{alg:hybrid-ltl} learns an LTL\textsubscript{f} formula $\varphi$ that is an $\nrep$ and for every LTL\textsubscript{f} formulas $\varphi'$ that is an $\nrep$, $\varphi'\not\rightarrow\varphi$ or $\varphi\rightarrow\varphi'$.
\end{theorem}

We now focus on the construction of $\ltlencoding$, which is significantly different from that of $\dfaencoding$. It is defined as follows:
\begin{align}\label{eq:ltl-encoding}
\ltlencoding \coloneqq \Psi_{\mathtt{LTL}}\wedge\Psi_{\positives}\wedge\Psi_{\negatives}\wedge \Psi_{\not\leftarrow\varphi}.
\end{align}
The first conjunct $\Psi_{\mathtt{LTL}}$ ensures that propositional variables we exploit encode a valid LTL\textsubscript{f} formula $\varphi'$.
The second conjunct $\Psi_{\positives}$ ensures that $\varphi'$ holds on all positive words, while the third, $\Psi_{\negatives}$, ensures that it does not hold on the negative words.
The final conjunct $\Phi_{\not\leftarrow\varphi}$ ensures that $\varphi\not\rightarrow \varphi'$.

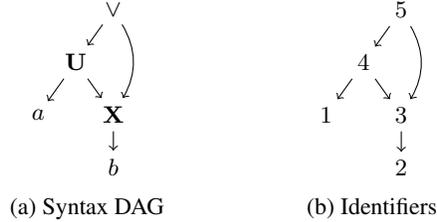
\begin{figure}
	\centering
	\subfloat[][Syntax DAG\label{fig:syntax-dag}]
	{\scalebox{1}{
			\begin{tikzpicture}
			\node (1) at (0, 0) {$\lor$};
			\node (2) at (-.5, -.7) {$\lU$};
			\node (3) at (-1.0, -1.4) {$a$};
			\node (4) at (0, -1.4) {$\lX$};
			\node (5) at (0, -2.1) {$b$};
			\node (6) at (-1.5,-.7) {};
			\node (7) at (0.8,-.7) {};
			\draw[->] (1) -- (2); 
			\draw[->] (1) to [bend left] (4);
			\draw[->] (2) -- (3);
			\draw[->] (2) -- (4);
			\draw[->] (4) -- (5);
			\end{tikzpicture}}
	}
	\hspace{1cm}
	\subfloat[][Identifiers\label{fig:identifier}]{
		\scalebox{1}{
			\begin{tikzpicture}
			\node (5) at (0, 0) {$5$};
			\node (4) at (-.5, -.7) {$4$};
			\node (3) at (0, -1.4) {$3$};
			\node (2) at (0, -2.1) {$2$};
			\node (1) at (-1.0, -1.4) {$1$};
			\node (6) at (-1.5,-.7) {};
			\node (7) at (0.8,-.7) {};
			\draw[->] (5) -- (4); 
			\draw[->] (5) to [bend left] (3);
			\draw[->] (4) -- (3);
			\draw[->] (4) -- (1);
			\draw[->] (3) -- (2);
			\end{tikzpicture}}
	}
	\caption{Syntax DAG and identifiers for $(a\lU\lX b) \vee \lX b$}
	\label{fig:representation}
\end{figure}
Following \citet{NeiderG18}, all of our conjuncts rely on a canonical syntactic representation of LTL\textsubscript{f} formulas as \emph{syntax DAGs}. A syntax DAG is a directed acyclic graph (DAG) that is obtained from the syntax tree of an LTL\textsubscript{f} formula by merging its common subformulas. 
An example of a syntax DAG is illustrated in Figure~\ref{fig:syntax-dag}.
Further, to uniquely identify each node of a syntax DAG, we assign them unique identifiers from $\upto{n}$ such that every parent node has an identifier larger than its children (see Figure~\ref{fig:identifier}).

To construct the hypothesis $\varphi'$, we encode its syntax DAG, using the following propositional variables: 
\renewcommand\labelenumi{(\theenumi)}
\begin{enumerate*}
	\item $x_{i,\lambda}$ for $i\in\upto{n}$ and $\lambda\in\Lambda$; and \item $l_{i,j}$ and $r_{i,j}$ for $i\in\upto{n}$ and $j\in\upto{i{-}1}$.
\end{enumerate*}
The variable $x_{i,\lambda}$ tracks the operator label of the Node~$i$ of the syntax DAG of $\varphi'$, while variables $l_{i,j}$ and $r_{i,j}$ encode the left and right child of Node~$i$, respectively.
Mathematically, $x_{i,\lambda}$ is set to true if and only if Node~$i$ is labeled with operator $\lambda$.
Moreover, $l_{i,j}$ (resp. $r_{i,j}$) is set to true if and only if  Node~$i$'s left (resp. right) child is Node~$j$.

To ensure variables $x_{i,\lambda}$, $l_{i,j}$ and $r_{i,j}$ have the desired meaning, $\Psi_{\mathtt{LTL}}$ imposes certain structural constraints, which we list below.
\begin{align}
\bigwedge_{i\in\upto{\bound}}\Big[\bigvee_{\lambda\in\Lambda} x_{i,\lambda} \wedge \bigwedge_{\lambda\neq \lambda'\in\Lambda} \big[\neg x_{i,\lambda} \vee \neg x_{i,\lambda'}\big]\Big]\label{eq:op-uniqueness}\\
\bigwedge_{i\in\{2,\cdots,\bound\}}\Big[\bigvee_{j\in\upto{i-1}} l_{i,j} \wedge \bigwedge_{j\neq j'\in\upto{i-1}} \big[\neg r_{i,j} \vee \neg l_{i,j'}\big]\Big]\label{eq:left-child-uniqueness}\\
\bigwedge_{i\in\{2,\cdots,\bound\}}\Big[\bigvee_{j\in\upto{i-1}} r_{i,j} \wedge \bigwedge_{j\neq j'\in\upto{i-1}} \big[\neg r_{i,j} \vee \neg r_{i,j'}\big]\Big]\label{eq:right-child-uniqueness}\\
\bigvee_{a\in\Sigma} x_{1,a}\label{eq:first-node-symbol}
\end{align}
Constraint~\ref{eq:op-uniqueness} ensures that each node of the syntax DAG of $\varphi'$ is uniquely labeled by an operator.
Constraints~\ref{eq:left-child-uniqueness} and~\ref{eq:right-child-uniqueness} encode that each node of the syntax DAG of $\varphi'$ has a unique left and right child, respectively.
Finally, Constraint~\ref{eq:first-node-symbol} encodes that the node with identifier 1 must always be a symbol.

We now describe the construction of $\Psi_{\positives}$ and $\Psi_{\negatives}$.
Both use variables $y^i_{w,t}$ where $i\in\upto{\bound}$, $w\in\positives\cup\negatives$, and $t\in\upto{\abs{w}}$.
The variable $y^i_{w,t}$ tracks whether $\varphi'_i$ holds on the suffix $w[t{:}]$, where $\varphi'_i$ is the subformula of $\varphi'$ rooted at Node~$i$.
Formally, $y^i_{w,t}$ is set to true if and only if $w[t{:}]\models\varphi'_i$.

To ensure the desired meaning of variables $y^{i}_{w,t}$, $\Psi_{\positives}$ and $\Psi_{\negatives}$ impose semantic constraints, again to ones proposed by~\citet{NeiderG18}.
\begin{align}
\bigwedge\limits_{i\in\upto{n}}\bigwedge\limits_{a\in \Sigma} x_{i,a}\rightarrow\Big[\bigwedge\limits_{t\in\upto{\abs{w}}}
\begin{cases}
y^{i}_{w,t}\text{ if }a = w[t] \\
\neg y^{i}_{w,t}\text{ if }a \neq w[t]
\end{cases}\Big]\label{eq:ltl-prop}\\
\bigwedge\limits_{\substack{i\in\upto{n}\\ j,j'\in\upto{i-1}}}x_{i,\vee}\wedge l_{i,j}\wedge r_{i,j'}\rightarrow\Big[\bigwedge\limits_{\substack{\tau\in\upto{\abs{w}}}}\Big[y^{i}_{w,t}\leftrightarrow y^{j}_{w,t}\vee y^{j'}_{w, t}\Big]\Big]\label{eq:ltl-or}\\
\bigwedge\limits_{\substack{i\in \upto{n}\\j\in\upto{i-1}}} x_{i,\lX} \land l_{i,j}\rightarrow \big[\bigwedge_{t\in\upto{\abs{w}-1}} y^{i}_{w,t} \leftrightarrow y^{j}_{w,t+1}\big] \wedge \neg y^{j}_{w,\abs{w}}\label{eq:ltl-next}\\
\bigwedge\limits_{\substack{i \in\upto{\bound} \\ j, j'\in\upto{i}}}x_{i,\lU}\wedge l_{i,j}\wedge r_{i,j'}\rightarrow\Big[\bigwedge\limits_{\substack{t\in\upto{ \abs{w}}}}\Big[y^{i}_{w,t}\leftrightarrow 
\bigvee\limits_{t\leq t'\in\upto{\abs{w}}}\Big[y^{j'}_{w,t'}\wedge  \bigwedge\limits_{t\leq \tau< t'} y^{j}_{w,\tau}\Big]\Big]\label{eq:ltl-until}
\end{align}
Intuitively, the above constraints encode the meaning of the different LTL operators using propositional logic.

To ensure that $\varphi'$ holds on positive words, we have $\Psi_{\positives}\coloneqq \bigwedge_{w\in\positives} y^{\bound}_{w,0}$ and to ensure $\varphi'$ does not hold on negative words, we have $\Psi_{\negatives}\coloneqq\bigwedge_{w\in\negatives} \neg y^{\bound}_{w,0}$.

Next, to construct $\Psi_{\not\leftarrow\varphi}$, we symbolically encode a word $u$  that distinguishes formulas $\varphi$ and $\varphi'$. 
We bound the length of the symbolic word by a \emph{time horizon} $\horizon=2^{2^{\bound+1}}$.
The choice of $\horizon$ is derived from Lemma~\ref{lem:time-horizon} and the fact that the size of the equivalent DFA for an LTL\textsubscript{f} formulas can be at most doubly exponential~\citep{DBLP:conf/ijcai/GiacomoV15}.

Our encoding of a symbolic word $u$ relies on variables $p_{t,a}$ where $t\in\upto{\horizon}$ and $a\in\Sigma\cup\{\varepsilon\}$.
If $p_{t,a}$ is set to true, then $u[t]=a$. To ensure that the variables $p_{t,a}$ encode their desired meaning, we generate a formula $\Psi_{\mathtt{word}}$ that consists of the following constraint:
\begin{align}
\bigwedge_{t\in\upto{\horizon}}\Big[\bigvee_{a\in\Sigma\cup\{\varepsilon\}} p_{t,a} \wedge \bigwedge_{a\neq a'\in\Sigma\cup\{\varepsilon\}} \big[\neg p_{t,a} \vee \neg p_{t,a'}\big]\Big]\label{eq:symbol-uniqueness}
\end{align}
The above constraint ensures that, in the word $u$, each position $t\leq\horizon$ has a unique symbol from $\Sigma\cup\{\varepsilon\}$.

Further, to track whether $\varphi$ and $\varphi'$ hold on $u$, we have variables $z^{\varphi,i}_{u,t}$ and $z^{\varphi',i}_{u,t}$ where $i\in\upto{n}$, $t\in\upto{\horizon}$.
These variables are similar to $y^i_{w,t}$, in the sense that $z^{\varphi,i}_{u,t}$ (resp. $z^{\varphi',i}_{u,t}$) is set to true, if $\varphi$ (resp. $\varphi'$) holds at position $t$.
To ensure desired meaning of these variables, we impose semantic constraints $\Psi_{\mathtt{sem}}$, similar to the semantic constraints imposed on $y^{i}_{w,t}$.
We list the constraints on the variables $z^{\varphi',i}_{u,t}$.
\begin{align}
\bigwedge\limits_{i\in\upto{n}}\bigwedge\limits_{a\in \Sigma} x_{i,a}\rightarrow\Big[\bigwedge\limits_{t\in\upto{\horizon}}
z^{\varphi',i}_{u,t} \leftrightarrow p_{t,a}\Big]\\
\bigwedge\limits_{\substack{i\in\upto{n}\\ j,j'\in\upto{i-1}}}x_{i,\vee}\wedge l_{i,j}\wedge r_{i,j'}\rightarrow\Big[\bigwedge\limits_{\substack{t\in\upto{\horizon}}}\Big[z^{\varphi',i}_{w,t}\leftrightarrow z^{\varphi',j}_{u,t}\vee z^{\varphi',j'}_{u, t}\Big]\Big]\\
\bigwedge\limits_{\substack{i\in \upto{n}\\j\in\upto{i-1}}} x_{i,\lX} \land l_{i,j}\rightarrow\big[\bigwedge_{t\in\upto{\horizon-1}} z^{\varphi',i}_{w,t} \leftrightarrow z^{\varphi',j}_{u,t+1}\big] 
\wedge \neg z^{\varphi',j}_{u,\abs{w}}\\
\bigwedge\limits_{\substack{i \in\upto{\bound} \\ j, j'\in\upto{i}}}x_{i,\lU}\wedge l_{i,j}\wedge r_{i,j'}\rightarrow\Big[\bigwedge\limits_{\substack{t\in\upto{ \horizon}}}\Big[z^{\varphi',i}_{u,t}\leftrightarrow
\bigvee\limits_{t\leq t'\in\upto{\horizon}}\Big[z^{\varphi',j'}_{u,t'}\wedge  \bigwedge\limits_{t\leq \tau< t'} z^{\varphi',j}_{u,\tau}\Big]\Big]
\end{align}
Clearly, the above constraints are identical to the Constraints~\ref{eq:ltl-prop} to~\ref{eq:ltl-until} imposed on variables $y^{i}_{w,t}$.
The constraints on variables $z^{\varphi,i}_{w,t}$ are quite similar.
The only difference is that for hypothesis $\varphi$, we know the syntax DAG exactly and thus, can exploit it instead of using an encoding of the syntax DAG.

Finally, we set $\Psi_{\not\leftarrow\varphi}\coloneqq\Psi_{\mathtt{word}}\wedge\Psi_{\mathtt{sem}}\wedge z^{\varphi,\bound}_{u,0}\wedge \neg z^{\varphi',\bound}_{u,0}$.
Intuitively, the above conjunction ensures that there exists a word on which $\varphi$ holds and $\varphi'$ does not. 

We now prove the correctness of the encoding $\ltlencoding$ described using the constraints above.
\begin{theorem}\label{thm:ltl-encoding-correctness}
	Let $\ltlencoding$ be as defined above. Then, we have the following:
	\begin{enumerate}
		\item If $\ltlencoding$ is satisfiable, then there exists an LTL formula $\varphi'$ that is an $\nrep$, $\varphi'$ does not hold on $w\in\negatives$ and $\varphi\not\rightarrow\varphi'$.
		\item If there exists a LTL formula $\varphi'$ that is an $\nrep$, $\varphi'$ does not hold on $w\in\negatives$ and $\varphi\not\rightarrow\varphi'$, then $\ltlencoding$ is satisfiable. 
	\end{enumerate}
\end{theorem}
To prove this theorem, we rely on intermediate claims, which we state now.
In the claims, $v$ is a model of $\ltlencoding$, $\varphi'$ is the LTL formula constructed from $v$ and $\varphi$ is the current hypothesis LTL formula.
\begin{numclaim}\label{cl:ltl-positive}
	For all $w\in\positives\cup\negatives$, $v(y^{i}_{w,t})=1$ if and only if $w[t{:}]\models\varphi'[i]$. 
\end{numclaim}
The proof proceeds via structural induction over $\varphi'[i]$.
For the proof of this claim, we refer the readers to the correctness proof of the encoding used by~\citep{NeiderG18} that can in found in their full paper~\citep{ng-arxiv-v1}.

\begin{numclaim}\label{cl:ltl-ineq}
	$v(z^{\varphi',i}_{w,t})=1$ (resp. $v(z^{\varphi,i}_{w,t})=1$) if and only for a word $w$, $w[t{:}]\models \varphi'$ (resp. $w[t{:}]\models \varphi$).
\end{numclaim}
The proof again proceeds via a structural induction on $\varphi'$ (similar to the previous one).

We are now ready to prove Theorem~\ref{thm:ltl-encoding-correctness}, i.e., the correctness of $\ltlencodingce$.

\begin{proof}[of Theorem~\ref{thm:ltl-encoding-correctness}]
	For the forward direction, consider that $\ltlencoding$ is satisfiable with a model $v$ and $\varphi'$ is the LTL formula constructed using the model $v$.
	First, using Claim~\ref{cl:ltl-positive}, we have that $v(y^{n}_{w,t})=1$ if and only if $w[t{:}]\models\varphi'$.
	Now, based on the constraints $\Psi_{\positives}$ and $\Psi_{\negatives}$, we observe that $v(y^{n}_{w,0})=1$ for all words $w\in\positives$ and $v(y^{n}_{w,0})=0$ for all words $w\in\negatives$. 
	Thus, combining the two observations, we conclude $w\models\varphi'$ for $w\in\positives$ and $w\not\models\varphi'$ for $w\in\negatives$ and hence, $\varphi'$ is an $\nrep$.
	Next, using Claim~\ref{cl:ltl-ineq} and conjunct $\Psi_{\not\leftarrow\varphi}$, we conclude that there exists a word $w$ on which $\varphi$ holds and $\varphi'$ does not.
	Thus, in total, we obtain $\varphi'$ to be an $\nrep$ which does not hold on $w\in\negatives$ and $\varphi\not\rightarrow\varphi'$ 
	
	For the other direction, based on a suitable LTL formula $\varphi'$, we construct an assignment $v$ for all the introduce propositional variables. 
	First, we set $v(x_{i,\lambda})=1$ if Node~$i$ is labeled with operator $\lambda$ and $v(l_{i,j})=1$ (resp. $v(r_{i,j})=1$) if left (resp. right) child of Node~$i$ is Node~$j$.
	Since $\varphi$ is a valid LTL formula, it is clear that the structural constraints will be satisfied by $v$.
	Similarly, we set $v(y^{i}_{w,t})=1$ if $w[t{:}]\models\varphi'[i]$ for some $w\in\positives\cup\negatives$, $t\in\upto{\abs{w}}$.
	It is again a simple exercise to check that $v$ satisfies Constraints~\ref{eq:ltl-prop} to~\ref{eq:ltl-until}.
	Next, we set $v(z^{\varphi}_{w,t})=1$ and $v(z^{\varphi',n}_{w,t})=0$ for a word $w$ for which $w\models\varphi$ and $w\not\models\varphi'$.
	It is again easy to check that $v$ satisfies $\Psi_{\not{\leftarrow}\varphi}$.
\end{proof}

We now prove the termination and correctness (Theorem~\ref{thm:semi-symbolic-ltl-correctness}) of Algorithm~\ref{alg:hybrid-ltl}.

\begin{proof}[of Theorem~\ref{thm:semi-symbolic-ltl-correctness}]
First, observe that Algorithm~\ref{alg:hybrid-ltl} always terminates. 
This is because, there are only finitely many LTL formulas with a given size bound $\bound$ and, as we show next, the algorithm produces a new LTL formula as a hypothesis in each iteration.
To show that new LTL formula are produced, assume that the algorithm finds same hypothesis in the iterations $k$ and $l$ (say $i<j$) of its (while) loop.
For clarity, let the hypothesis produced in an iteration $i$ be indicated using $\varphi_i$.
We now argue using case analysis.
Consider the case where in all iterations between $i$ and $j$, the condition  $\varphi'\rightarrow\varphi$ (Line~\ref{line:subset}) holds.
In such a case, $\varphi_l\rightarrow\varphi_k$ and $\varphi_k\not\rightarrow\varphi_l$ and thus, $\varphi_l\neq \varphi_k$, contradicting the assumption.
Now, consider the other case where, in at least one of the iterations between $i$ and $j$, the (else) condition $\varphi'\not\rightarrow\varphi$ (Line~\ref{line:incomparable}) holds.
In such a case, a word $w$ on which $\varphi_k$ holds is added to the set of negative words $\negatives$.
Since $\varphi_l$ must not hold on this negative word $w$, $\varphi_l\neq \varphi_k$, again contradicting the assumption.

Next, we prove that the LTL formula $\varphi$ learned by the algorithm is the correct one, using contradiction.
We assume that there is an LTL formula $\overline{\varphi}$ that is an $\nrep$ and satisfies the conditions $\overline{\varphi}\rightarrow \varphi$ and $\varphi\not\rightarrow\overline{\varphi}$.
Now, observe that the algorithm terminates when $\ltlencoding$ is unsatisfiable.
Based on the properties of $\ltlencoding$ established in Theorem~\ref{thm:ltl-encoding-correctness}, if $\ltlencoding$ is unsatisfiable for some LTL formula $\varphi$ and set of negative words $\negatives$, then any LTL formula $\varphi'$ that is an $\nrep$ either it holds in one of the words in $\negatives$ or $\varphi\rightarrow\varphi'$.
If $\varphi\rightarrow\overline{\varphi}$, then the assumption $\varphi\not\rightarrow\overline{\varphi}$ is contradicted.
If $\overline{\varphi}$ holds in one of the negative words, the assumption $\overline{\varphi}\rightarrow \varphi$ is contradicted.
\end{proof}

\subsection{The Counterexample-guided Algorithm}\label{sec:ce-ltl}
We now design a \emph{counterexample-guided} algorithm to solve Problem~\ref{prob:occ-ltl}.
In contrast to the symbolic (or semi-symbolic) algorithm, this algorithm does not guide the search based on propositional formulas built out of the hypothesis LTL\textsubscript{f} formula.
Instead, this algorithm relies entirely on two sets: a set of negative words $\negatives$ and a set of discarded LTL\textsubscript{f} formulas $\discardpile$.
Based on these two sets, we design a propositional formula $\ltlencodingce$ that has the properties that:
\begin{enumerate*}
	\item $\ltlencodingce$ is satisfiable if and only if there exists an LTL\textsubscript{f} formula $\varphi$ that is an $\nrep$, does not hold on $w\in\negatives$, and is not one of the formulas in $D$; and 
	\item based on a model $\model$ of $\ltlencodingce$, one can construct such an LTL\textsubscript{f} formula $\varphi'$.
\end{enumerate*}

\begin{algorithm}[tb]
	\caption{CEG Algorithm for LTL\textsubscript{f} formulas }\label{alg:ce-ltl}
	\textbf{Input}: Positive words~$P$, bound~$\bound$
	\begin{algorithmic}[1]
		\STATE $N\gets\emptyset$, $D \gets \emptyset$ 
		\STATE $\varphi\gets\varphi_{\Sigma^\ast}$, $\ltlencodingce \coloneqq \Psi_{\mathtt{LTL}}\wedge\Psi_{\positives}$
		\WHILE{$\ltlencodingce$ is satisfiable (with model $\model$)}
		\STATE $\varphi' \gets \varphi^{v}$
		\IF{$\varphi'\leftrightarrow\varphi$}
		\STATE Add $\varphi'$ to $D$
		\ELSE			
		\IF{$\varphi'\rightarrow\varphi$}
		\STATE Add $w$ to $\negatives$, where  $w\models\neg\varphi\wedge\varphi'$
		\STATE $\varphi\gets\varphi'$
		\ELSE
		\STATE Add $w$ to $\negatives$, where $w\models\neg\varphi'\wedge\varphi$ 
		\ENDIF
		\ENDIF
		\STATE $\ltlencodingce \coloneqq \Psi_{\mathtt{LTL}}\wedge \Psi_{\positives}\wedge\Psi_{\negatives}\wedge\Psi_{D}$
		\ENDWHILE
		\RETURN $\varphi$
	\end{algorithmic}
\end{algorithm}

Being a counterexample-guided algorithm, the construction of the sets $\negatives$ and $\discardpile$ forms the crux of the algorithm.
In each iteration, these sets are updated based on the relation between the hypothesis $\varphi$ and the current guess $\varphi'$ (obtained from a model of $\ltlencodingce$).
There are exactly three relevant cases, which we discuss briefly.
\begin{itemize}
	\item First, $\varphi'\leftrightarrow\varphi$, i.e., $\varphi'$ and $\varphi$ hold on the exact same set of words. In this case, the algorithm discards $\varphi'$, due to its equivalence to $\varphi$, by adding it to $D$.
	\item Second, $\varphi'\rightarrow\varphi$ and $\varphi\not\rightarrow\varphi'$, i.e., $\varphi'$ holds on a proper subset of the set of words on which $\varphi$ hold. In this case, our algorithm generates a word that satisfies $\varphi$ and not $\varphi'$, which it adds to $\negatives$ to eliminate $\varphi$.
	\item Third, $\varphi'\not\leftarrow\varphi$, i.e., $\varphi'$ does not hold on a subset of the set of words on which $\varphi$ hold. In this case, our algorithm generates a word $w$ that satisfies $\varphi'$ and not $\varphi$, which it adds to $\negatives$ to eliminate $\varphi'$. 
\end{itemize}
By handling the cases mentioned above, we obtain an algorithm (sketched in Algorithm~\ref{alg:ce-ltl}) with guarantees (formalized in Theorem~\ref{thm:semi-symbolic-ltl-correctness}) exactly the same as the semi-symbolic algorithm in Section~\ref{sec:semi-sym-ltl}.

\section{Experiments}
\label{sec:experiments}

In this section, we evaluate the performance of the proposed algorithms using three case studies.
First, we evaluate the performance of Algorithm~\ref{alg:symbolic-dfa}, referred to as \symbolicDfa{}, and compare it to a baseline counterexample-guided algorithm by \citet{AvellanedaP18}, referred to as \ceDfa{} in our first case study.
Then, we evaluate the performance of the proposed semi-symbolic algorithm (Section~\ref{sec:semi-sym-ltl}), referred to as \hybridLtl{}, and the counterexample-guided algorithm (Section~\ref{sec:ce-ltl}), referred to as \ceLtl{}, for learning LTL\textsubscript{f} formulas in our second and third case studies.

\begin{figure}[h]
	\begin{center}
		\includegraphics[width=0.8\columnwidth]{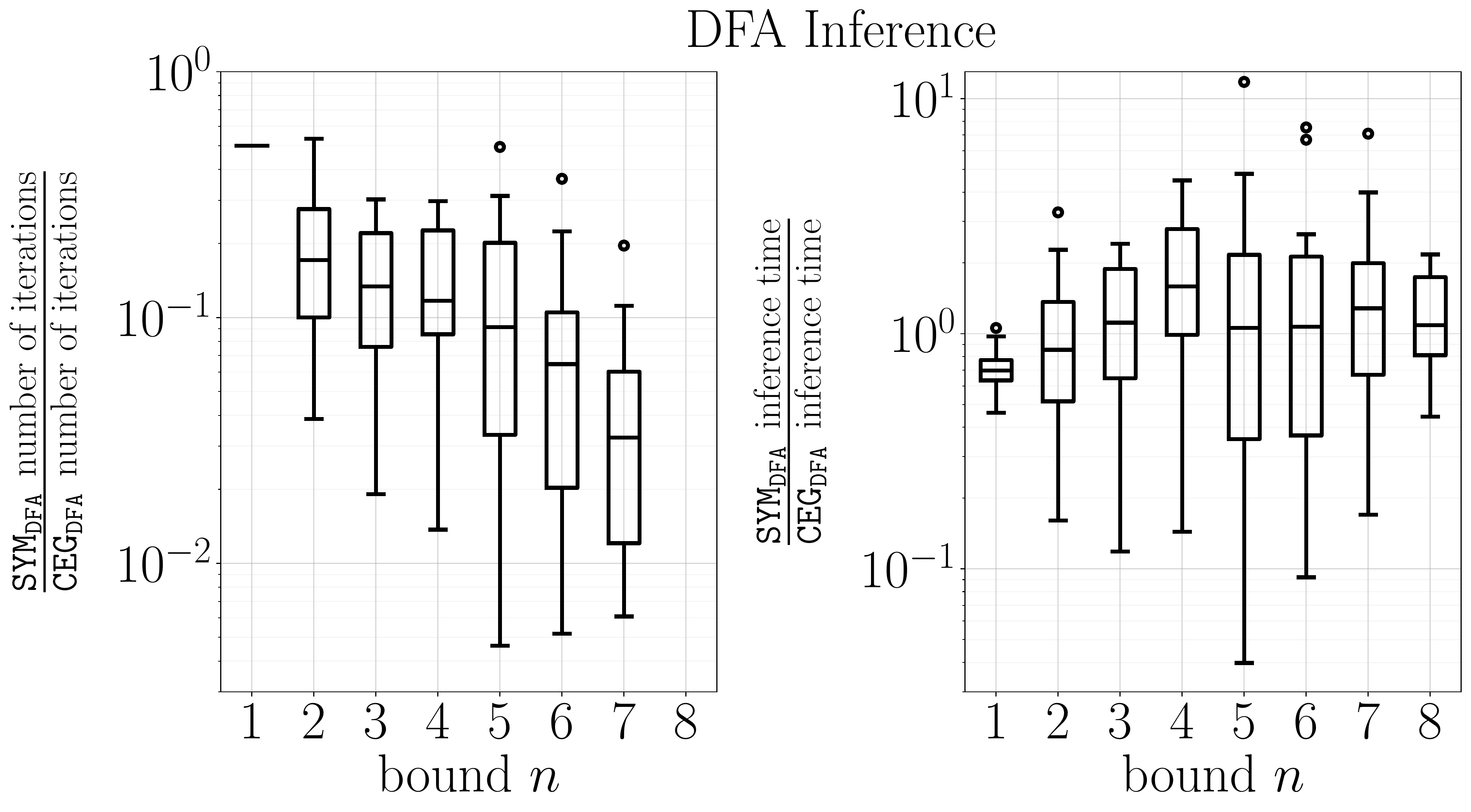}
		\caption{
			Comparison of \symbolicDfa{} and \ceDfa{} in terms of the runtime and the number of iterations of the main loop.
		}  
		\label{fig:DFA}
	\end{center}
\end{figure}

\begin{table*}[]
	\centering
	\begin{tabular}{c c c}
		\toprule
		absence
		&
		existence
		&
		universality
		\\ \hline
		$\lG(\lnot{a_0})$
		&
		$\lF{a_0}$
		&
		$\lG{a_0}$
		\\ 
		$\lG(a_1\rightarrow{\lG(\lnot{a_0})})$
		&
		$(\lG(\lnot{a_0}))\lor(\lF(a_0\land (\lF{a_1})))$
		&
		$\lG(a_1\rightarrow{\lG{a_0}})$
		\\ 
		$\lF{a_1}\rightarrow{(\lnot a_0\lU{a_1}})$
		&
		$\lG(a_0\land ((\lnot{a_1})\rightarrow((\lnot{a_1})\lU(a_2\land (\lnot{a_1})))))$
		&
		$\lF{a_1} \rightarrow (a_0 \lU{a_1})$
		
		\\ \midrule
		\multicolumn{3}{c}{disjunction of common patterns}
		\\ \hline
		\multicolumn{3}{c}{$(\lF{a_2})\lor((\lF{a_0})\lor(\lF{a_1}))$}
		\\ 
		\multicolumn{3}{c}{$((\lG(\lnot{a_0}))\lor{ (\lF(a_0\land(\lF{a_1}))))}\lor((\lG(\lnot{a_3})) \lor(\lF(a_2\land(\lF {a_3}))))$}
		\\ 
		\multicolumn{3}{c}{$(\lG(a_0\land((\lnot{a_1})\rightarrow((\lnot{a_1})\lU(a_2\land(\lnot{a_1}))))))\lor(\lG(a_3\land((\lnot{a_4})\rightarrow((\lnot{a_4})\lU(a_5\land(\lnot{a_4}))))))$}
		\\ \bottomrule
	\end{tabular}
	\caption{Common LTL patterns used for generation of words.}
	\label{tab:LTL_tab}
\end{table*}



In \hybridLtl{}, we fixed the time horizon $\horizon$ to a natural number, instead of the double exponential theoretical upper bound of $2^{2^{n+1}}$. 
Using this heuristic means that \hybridLtl{} does not solve Problem \ref{prob:occ-ltl}, but we demonstrate that we produced good enough formulas in practice.


In addition, we implemented two existing heuristics from \citet{AvellanedaP18} to all the algorithms.
First, in every algorithm, we learned models in an incremental manner, i.e., we started by learning DFAs (resp. LTL\textsubscript{f} formulas) of size 1 and then increased the size by 1. 
We repeated the process until bound $\bound$.
Second, we used a set of positive words~$P'$ instead of $P$ that starts as an empty set, 
and at each iteration of the algorithm, if the inferred language does not contain some words from  $P$, we then extended $P'$ with one of such words, preferably the shortest one.
This last heuristic helped when dealing with large input samples because it used as few words as possible from the positive examples $P$.

We implemented every algorithm in Python~3\footnote{\url{https://github.com/cryhot/samp2symb/tree/paper/posdata}},
using PySAT~\citep{imms-sat18} for learning DFA,
and an ASP~\citep{baral_2003} encoding that we solve using clingo~\citep{DBLP:journals/corr/GebserKKS17} for learning LTL\textsubscript{f} formulas.
Overall, we ran all the experiments using 8~GiB of RAM and two CPU cores  with clock speed of 3.6~GHz. 



\subsubsection{Learning DFAs}
\label{ssec:experiment:dfa-random}

For this case study, we considered a set of 28 random DFAs of size 2 to 10 generated using AALpy~\citep{AALpy}.
Using each random DFA, we generated a set of 1000 positive words of lengths 1 to 10. We ran algorithms \ceDfa{} and \symbolicDfa{} with a timeout  $TO=1000$s, and for $\bound$ up to 10.

Figure \ref{fig:DFA} shows a comparison between the performance of \symbolicDfa{} and \ceDfa{} in terms of the inference time and the required number of iterations of the main loop. On the left plot, the average ratio of the number of iterations is $0.14$, which, in fact, shows that \symbolicDfa{} required noticeably less number of iterations compared to \ceDfa{}. On the right plot, the average ratio of the inference time is $1.09$, which shows that the inference of the two algorithms is comparable, and yet \symbolicDfa{} is computationally less expensive since it requires fewer iterations.

\subsubsection{Learning Common LTL\textsubscript{f} Patterns}
\label{ssec:experiment:ltl-patterns}
In this case study, we generated sample words using 12 common LTL patterns~\citep{DwyerAC99}, which we list in Table~\ref{tab:LTL_tab}.
Using each of these 12 ground truth LTL\textsubscript{f} formulas, we generated a sample of 10000 positive words of length 10.
Then, we inferred LTL\textsubscript{f} formulas for each sample using \ceLtl{} and \hybridLtl{}, 
separately.
For both algorithms, we set the maximum formula size $\bound=10$ and a timeout of $TO=1000$s. For \hybridLtl{}, we additionally set the time horizon $\horizon=8$.

Figure~\ref{fig:LTLf} represents a comparison between the mentioned algorithms in terms of inference time for the ground truth LTL\textsubscript{f} formulas
$\psi_1$, $\psi_2$, and
$\psi_3$.
On average, \hybridLtl{} ran 173.9\% faster than \ceLtl{} for all the 12 samples.
Our results showed that the LTL\textsubscript{f} formulas $\varphi$ inferred by \hybridLtl{} were more or equally specific than the ground truth LTL\textsubscript{f} formulas $\psi$ (i.e., $\varphi\rightarrow\psi$) for five out of the 12 samples, while the LTL\textsubscript{f} formulas $\varphi'$ inferred by \ceLtl{} were equally or more specific than the ground truth LTL\textsubscript{f} formulas $\psi$ (i.e., $\varphi'\rightarrow\psi$) for three out of the 12 samples.

\begin{figure}
	\begin{center}
		\includegraphics[width=0.9\columnwidth]{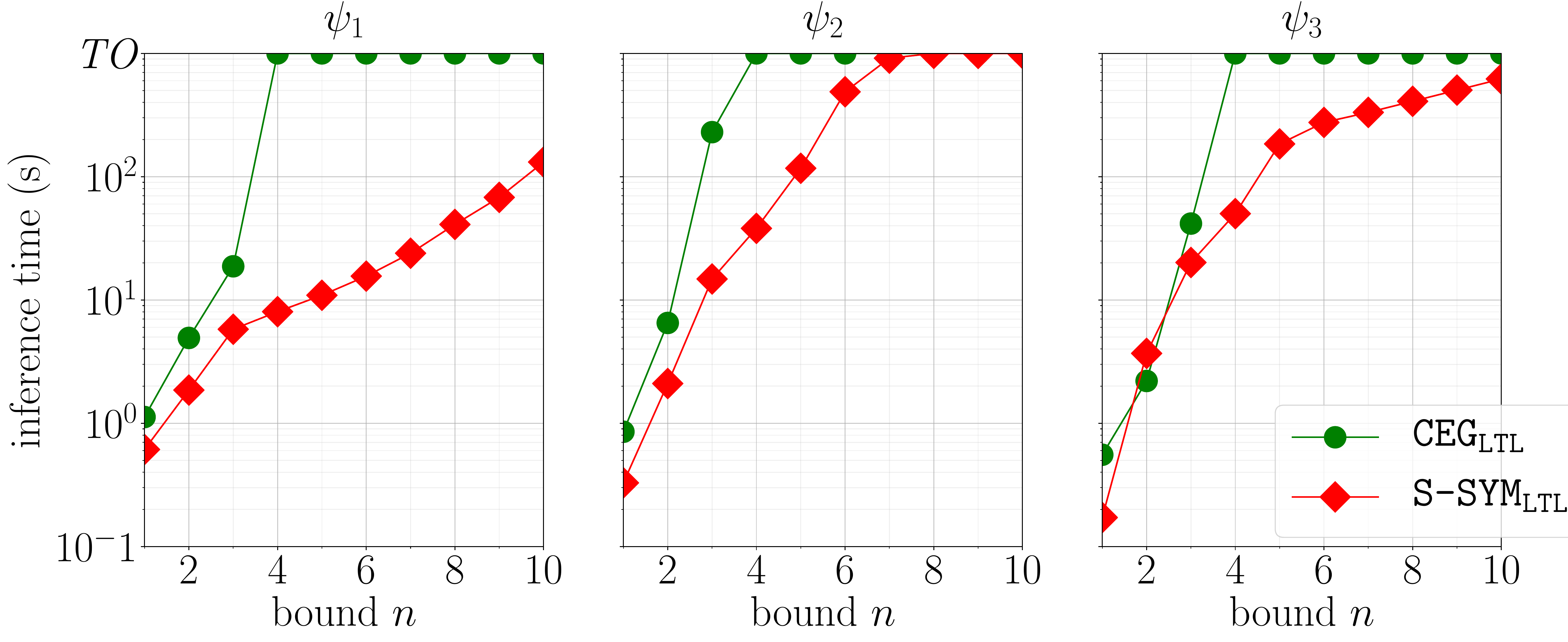}
		\caption{Comparison of \hybridLtl{} and \ceLtl{}{} in terms of the inference time for three LTL ground truth formulas.}
		\label{fig:LTLf}
	\end{center}
\end{figure}

\begin{figure}[h]
	\vspace{-0.4cm}
	\begin{center}
		\includegraphics[width=0.9\columnwidth]{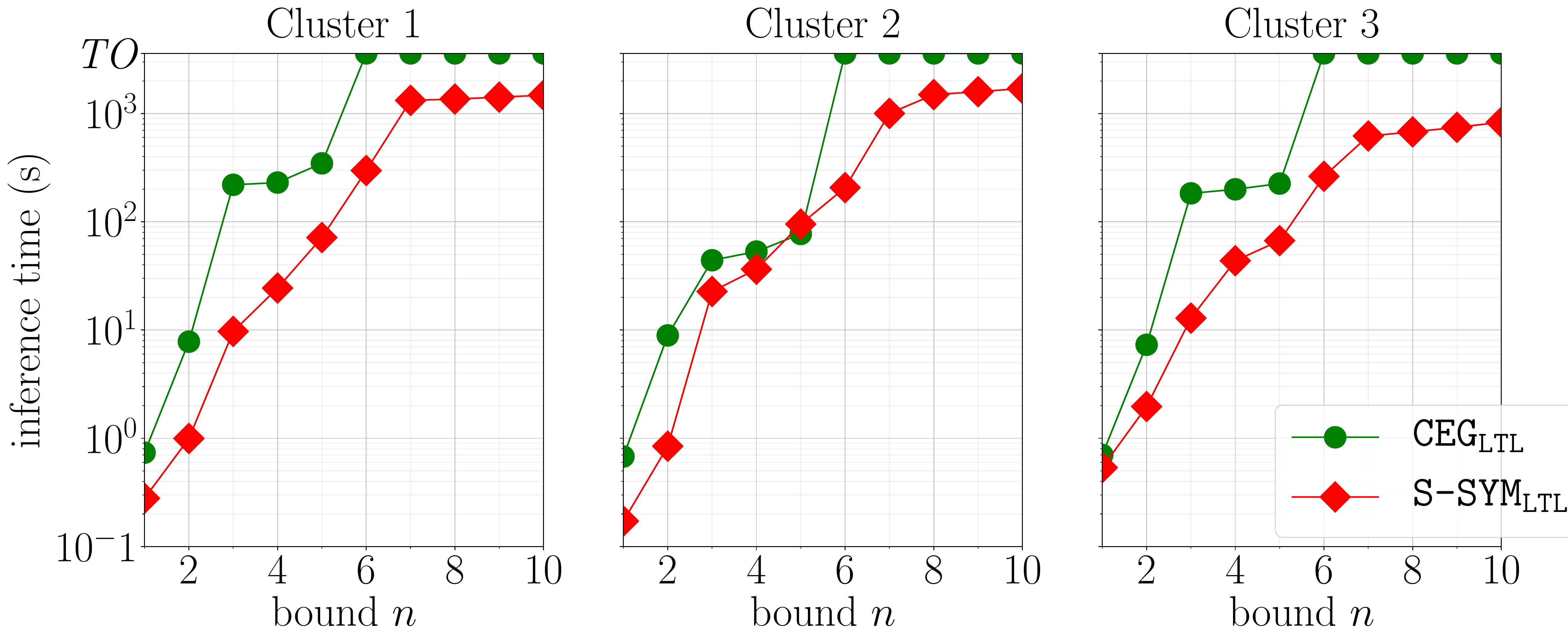}
		\caption{Comparison of \hybridLtl{} and \ceLtl{}{} in terms of the runtime for three clusters of words taken from a UAV.}  
		\label{fig:UAV}
	\end{center}
\end{figure}

\subsubsection{Learning LTL from Trajectories of Unmanned Aerial Vehicle (UAV)}
\label{ssec:experiment:ltl-uav}

In this case study, we implemented \hybridLtl{} and \ceLtl{} using sample words of a simulated unmanned aerial vehicle (UAV) for learning LTL\textsubscript{f} formulas.
Here, we used 10000 words clustered into three bundles using the $k$-means clustering approach.
Each word summarizes selective binary features such as $x_0$: ``low battery'', $x_1$: ``glide (not thrust)'', $x_2$: ``change yaw angle'', $x_3$: ``change roll angle'', etc.
We set $\bound=10$, $\horizon=8$, and a timeout of $TO=3600$s.
We inferred LTL\textsubscript{f} formulas for each cluster using \ceLtl{} and \hybridLtl{}.

Figure~\ref{fig:UAV} depicts a comparison between \ceLtl{} and \hybridLtl{} in terms of the inference time for three clusters.
Our results showed that, on average, \hybridLtl{} is 260.73\% faster than \ceLtl{}.
Two examples of the inferred LTL\textsubscript{f} formulas from the UAV words were $(\lF x_1)\limplies(\lG x_1)$ which reads as ``\textit{either the UAV always glides, or it never glides}'' and $\lG(x_2\limplies{x_3})$ which reads as ``\textit{a change in yaw angle is always accompanied by a change in roll angle}''.


\section{Conclusion}\label{sec:conclusion}
We presented novel algorithms for learning DFAs and LTL\textsubscript{f} formulas from positive examples only.
Our algorithms rely on conciseness and language minimality as regularizers to learn meaningful models.
We demonstrated the efficacy of our algorithms in three case studies.

A natural direction of future work is to lift our techniques to tackle learning from positive examples for other finite state machines (e.g., non-deterministic finite automata) and more expressive temporal logics (e.g., linear dynamic logic (LDL)~\citep{GiacomoV13}).

\section*{Acknowledgments}
We are especially grateful to Dhananjay Raju for introducing us to Answer Set Programming and guiding us in using it to solve our SAT problem.
This work has been supported by
the Defense Advanced Research Projects Agency (DARPA) (Contract number HR001120C0032),
Army Research Laboratory (ARL) (Contract number W911NF2020132 and ACC-APG-RTP W911NF),
National Science Foundation (NSF) (Contract number 1646522), and
Deutsche Forschungsgemeinschaft (DFG) (Grant number 434592664).

%

\bibliographystyle{plainnat}
\bibliography{bib.bib}

\appendix

\section{The Symbolic algorithm for learning DFAs with heuristics}

We present the complete symbolic algorithm for learning DFAs along with the main heuristics (Problem~\ref{prob:occ-dfa}).
The pseudocode is sketched in Algorithm~\ref{alg:symbolic-dfa-full}.
Compared to the algorithm presented in Algorithm~\ref{alg:symbolic-dfa}, we make a few modifications to improve performance.
First, we introduce a set $\positives'$ (also described in Section~\ref{sec:experiments}) to store the set of positive words necessary for learning the hypothesis DFA $\dfa'$.
Second, we incorporate an incremental DFA learning, meaning that we search for DFAs satisfying the propositional formula $\dfaencoding$ of increasing size (starting from size 1).
To reflect this, we extend $\dfaencoding$ with the size parameter, represented using $\Phi^{\dfa,m}$, to search for a DFA of size $m$.

\begin{algorithm}[h]
    \caption{Symbolic Algorithm for Learning DFA (with heuristics)}\label{alg:symbolic-dfa-full}
	\textbf{Input}: Positive words~$P$, bound~$\bound$
	\begin{algorithmic}[1]
		\STATE $\dfa\gets\dfa_{\Sigma^\ast}$,
		$\dfaencoding \gets \Phi_{\mathtt{DFA}}\wedge\Phi_{\positives}$
		\STATE $\positives'\gets \emptyset$,
		$m \gets 1$
		\WHILE{$m\leq\bound$}
            \STATE $\Phi^{\dfa,m} \gets \Phi^{m}_{\mathtt{DFA}}\wedge\Phi_{\positives'} \wedge \Phi_{\subseteq\dfa} \wedge \Phi_{\not\supseteq\dfa}$
            \IF{no model $\model$ satisfies $\Phi^{\dfa,m}$}
		        \STATE $m \gets m+1$
		    \ELSE
		        \STATE $\dfa' \gets$ DFA constructed from $v$
                \IF{exists $w \in \positives \setminus L(\dfa')$}
                    \STATE Add the shortest of such $w$ to $\positives'$
                \ELSE
                    \STATE $\dfa \gets \dfa'$
		        \ENDIF
		    \ENDIF
		\ENDWHILE
		\RETURN $\dfa$
	\end{algorithmic}
\end{algorithm}

\section{Comparison of Symbolic and Semi-symbolic Algorithm for Learning DFAs}

We have introduced counterexample-guided, semi-symbolic and symbolic approaches in this paper.
Our exploration of these methods will not be complete if we did not try a semi-symbolic algorithm for learning DFA.
Hence, we introduce \hybridDfa{}, a semi-symbolic approach for learning DFA.
This is done in a similar fashion than for LTL (Algorithm \ref{alg:hybrid-ltl}), but for DFA instead. Hence, we use the encoding $\dfaencoding \coloneqq \Phi_{\mathtt{DFA}}\wedge\Phi_{\positives} \wedge \Phi_{\negatives} \wedge \Phi_{\not\supseteq\dfa}$.
In practice, \hybridDfa{} is always worse than \symbolicDfa{},
both in term of inference time (in average, $3.2$ times more)
and number of iterations (in average, $2.1$ times more),
as demonstrated in Figure \ref{apdx:fig:DFA}.

\begin{figure}[h]
	\begin{center}
	\includegraphics[width=0.8\columnwidth]{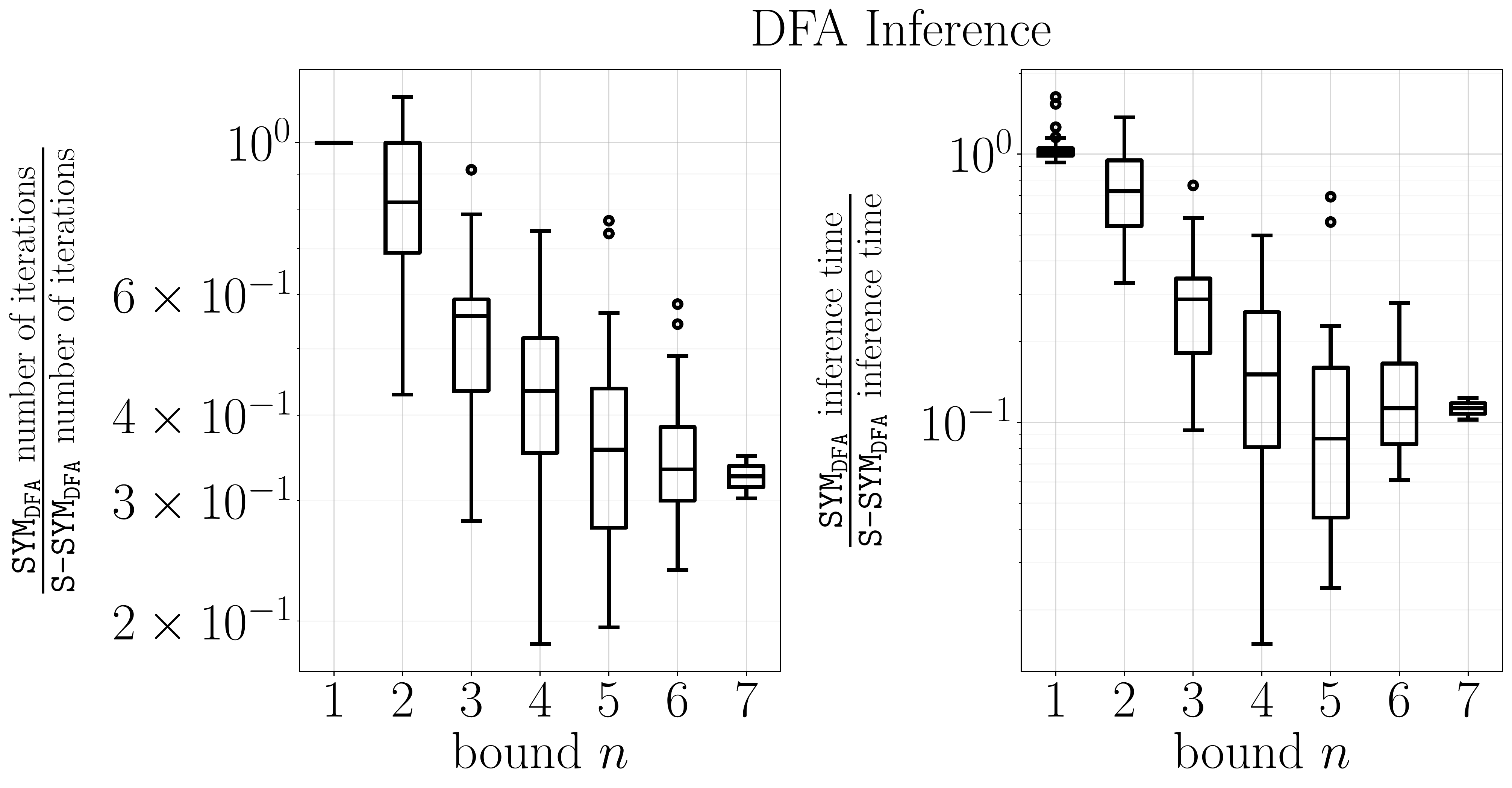}
	\caption{
	Comparison of \symbolicDfa{} and \hybridDfa{} in terms of the inference time and the number of iterations of the main loop.
	}  
	\label{apdx:fig:DFA}
	\end{center}
\end{figure}

\section{A Symbolic Algorithm for learning LTL formulas}
We now describe few modifications to the semi-symbolic algorithm presented in Section~\ref{sec:semi-sym-ltl} to convert it into a \emph{completely} symbolic approach.
This algorithm relies entirely on the hypothesis LTL formula $\varphi$ for constructing a propositional formula $\Psi^{\varphi}$ that guides the search of the next hypothesis.
Precisely, the formula $\Psi^{\varphi}$ has the properties that:
\begin{enumerate*}
\item $\Psi^{\varphi}$ is satisfiable if and only if there exists an LTL formula $\varphi'$ that is an $\nrep$ and $\varphi\not\rightarrow\varphi'$ and $\varphi'\rightarrow\varphi$; and
\item based on a model $\model$ of $\ltlencoding$, one can construct such an LTL formula.
\end{enumerate*}

The algorithm, sketched in Algorithm~\ref{alg:symbolic-ltl}, follows the same framework as Algorithm~\ref{alg:symbolic-dfa}. 
We here make necessary modifications to search for an LTL formula.
Also, the propositional formula $\Psi^{\varphi}$ has a construction similar to $\ltlencodingce$, with the exception that $\Psi_{\negatives}$ is replaced by $\Psi_{\rightarrow \varphi}$.
\begin{algorithm}[tb]
	\caption{Symbolic Algorithm for Learning LTL}\label{alg:symbolic-ltl}
	\textbf{Input}: Positive words~$P$, bound~$\bound$
	\begin{algorithmic}[1]
		\STATE $\varphi\gets\varphi_{\Sigma^\ast}$, $\Psi^{\varphi} \coloneqq \Psi_{\mathtt{LTL}}\wedge\Psi_{\positives}$

		\WHILE{$\Psi^{\varphi}$ is satisfiable (with model $\model$)}
		\STATE $\varphi \gets$ LTL formula constructed from $v$
		\STATE $\Psi^\varphi \coloneqq \Psi_{\mathtt{LTL}}\wedge\Psi_{\positives} \wedge \Psi_{\rightarrow\varphi} \wedge \Psi_{\not\leftarrow\varphi}$
		\ENDWHILE
		\RETURN $\varphi$
	\end{algorithmic}
\end{algorithm}

We here only describe the construction of the conjunct $\Psi_{\rightarrow\varphi}$ which reuses the variables and constraints already introduced in Section~\ref{sec:semi-sym-ltl}.
\begin{align}
    \Psi_{\rightarrow\varphi}\coloneqq\forall_{t\in\upto{\horizon},a\in{\Sigma}} p_{t,a}: \Big[\big[\Psi_{\mathtt{word}}\wedge \Psi_{\mathtt{sem}}\big] \rightarrow z^{\varphi',n}_{w,t}\rightarrow z^{\varphi,n}_{w,t}\Big]
\end{align}
Intuitively, the above constraint says that if for all words $u$ of length $\leq\horizon$, if $\varphi'$ holds on $u$, then so must $\varphi$.

\section{Evaluation of the Symbolic Algorithm for learning LTL formulas}

We refer to this symbolic algorithm for learning LTL formulas (Algorithm~\ref{alg:symbolic-ltl}) as \symbolicLtl{}.
We implement \symbolicLtl{} using QASP2QBF~\cite{https://doi.org/10.48550/arxiv.2108.06405}.
\symbolicLtl{} has an inference time several orders of magnitude above the inference time of \hybridLtl{}, as demonstrated in Figure \ref{apdx:fig:LTLf}.
This can be explained by the choice of the solver, and the inherent complexity of the problem due to quantifiers.
On the third experiment (Section \ref{ssec:experiment:ltl-uav}), \symbolicLtl{} timed out even for $\bound=1$.

\begin{figure}
	\begin{center}
	\includegraphics[width=0.9\columnwidth]{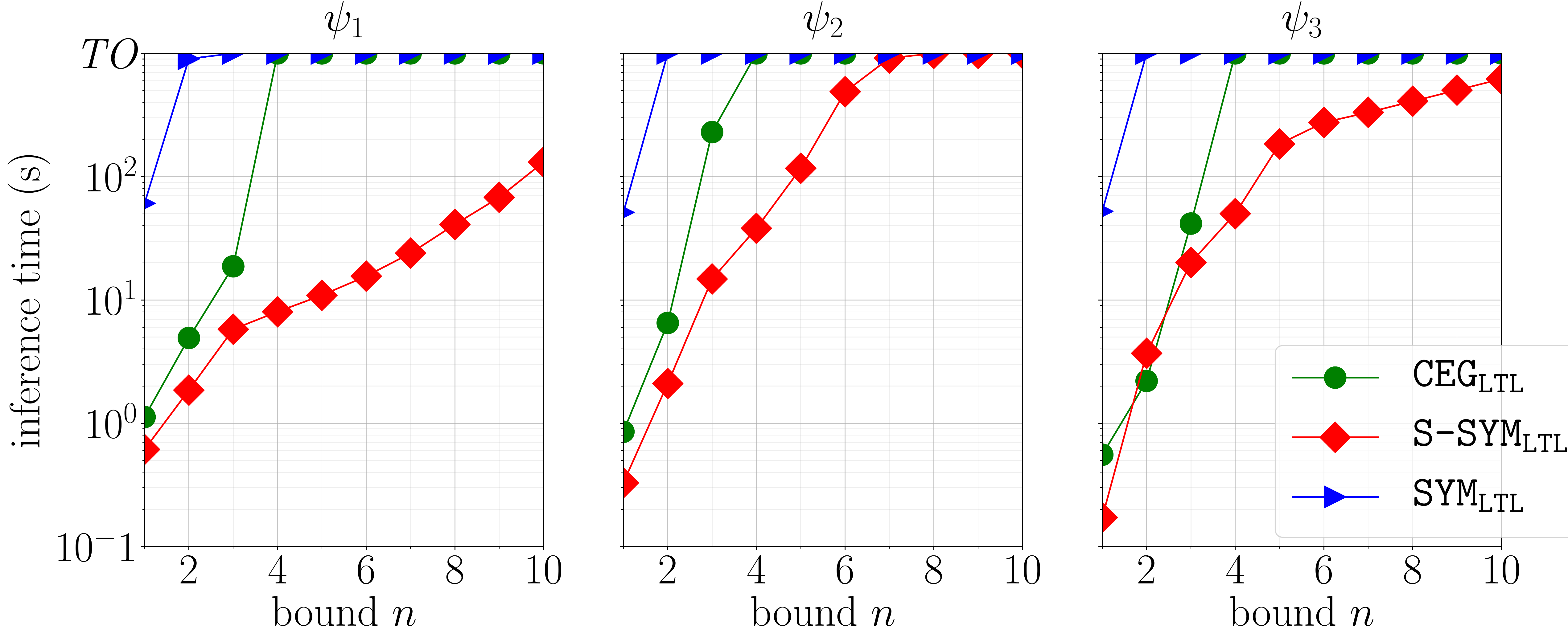}
	\caption{
	For each sample of the second experiment (Section \ref{ssec:experiment:ltl-patterns}), comparison of the inference time between \ceLtl{}, \hybridLtl{} and \symbolicLtl{}.
	}
	\label{apdx:fig:LTLf}
	\end{center}
\end{figure}

\end{document}